\documentclass[final]{siamltex}

\usepackage{amsmath,amssymb,epsfig,subfigure,overpic}
\usepackage{framed,color}

\newcommand{\dx}{\mathrm{d}x}

\definecolor{forestgreen}{rgb}{0.13, 0.55, 0.13}
\definecolor{maroon}{rgb}{0.75,0.19,0.38}
\definecolor{grey}{rgb}{0.6, 0.6, 0.6}

\title{Multiscale reaction-diffusion algorithms: \\
PDE-assisted Brownian dynamics}

\author{Benjamin Franz$^1$\!\!\! \and  
        Mark B. Flegg$^1$\!\!\! \and 
        S.~Jonathan Chapman$^1$\!\!\! \and \;\\ \quad
	Radek Erban$^{1,2}$}

\begin{document}

\maketitle

\footnotetext[1]{Mathematical Institute, 
	University of Oxford, 24-29 St Giles', Oxford OX1 3LB, United Kingdom}
\footnotetext[2]{Corresponding author: erban@maths.ox.ac.uk}
	
\begin{abstract}
    Two algorithms that combine Brownian dynamics (BD) simulations with
    mean-field partial differential equations (PDEs) are presented.
    This PDE-assisted Brownian dynamics (PBD) methodology provides exact 
    particle tracking data in parts of the 
    domain, whilst making use of a mean-field reaction-diffusion PDE 
    description elsewhere. The first PBD algorithm couples BD simulations
    with PDEs by  randomly creating  new particles close to the interface which 
    partitions the domain and by reincorporating particles into the 
    continuum PDE-description when they cross the interface. The second
    PBD algorithm introduces an overlap region, where both descriptions exist
    in parallel. It is shown that to accurately compute  variances
    using the PBD simulation requires the overlap region. Advantages of both
    PBD approaches are discussed and illustrative numerical examples are 
    presented. 
\end{abstract}

\begin{keywords}
    reaction-diffusion systems, Brownian dynamics, multiscale simulation
\end{keywords}

\begin{AMS}
35R60, 60J65, 92C40
\end{AMS}

\pagestyle{myheadings}
\thispagestyle{plain}

\section{Introduction} \label{sec:introduction}
Spatial reaction-diffusion models have been widely used for 
the description of biological systems \cite{Murray:MB2}.
Often continuum approaches, written in the form of reaction-diffusion
partial differential equations (PDEs), are used
due to their simplicity and the vast number of ready-to-use 
numerical solvers. However, many biological effects cannot 
be fully described by  deterministic PDE-based models. 
This is because a deterministic model requires large 
copy numbers of molecules to minimize the relative fluctuation of the spatial 
concentration. If low copy numbers are present in a biological
system \cite{Lipkow:2005:SDP, Takahashi:2010:STC}, then
stochastic models such as  mesoscopic compartment-based algorithms
\cite{Hattne:2005:SRD,Engblom:2009:SSR} or trajectory tracking 
(Brownian dynamics) methods, may be deployed 
\cite{Andrews:2004:SSC,vanZon:2005:SBN}. 

In many situations individual trajectories are important only 
in certain parts of the domain, whilst in the remainder of the 
domain a coarser, less detailed, method can be used \cite{Flegg:2012:TRM}. 
This is the case, for example, in the modelling of ion-channels 
\cite{Moy:2000:TCT}. Ions pass through a channel in 
single file and an individual-based model has to be used
to accurately compute the discrete, stochastic, current in the 
channel \cite{Chen:2012:BDM}. The positions of individual ions are 
less important away from the channel where copy numbers may be very 
large (rendering a detailed Brownian dynamics description infeasible) 
\cite{Corry:2000:TCT}. Another example is the stochastic 
reaction-diffusion modelling of filopodia which are dynamic finger-like 
protrusions used by eukaryotic motile cells to probe their 
environment and help guide cell motility \cite{Zhuravlev:2009:MNC}.
These relatively small protrusions are connected to a larger cytosol 
compartment. If a modeller is interested to understand the dynamics
of filopodia, then there is a potential to decrease the computational
cost of simulations by using a coarser model in the cytosol.
In both examples, it is important to understand how models with
a different level of detail can be used in different parts of the
computational domain \cite{Flegg:2012:TRM}.

In this paper, we develop algorithms that calculate Brownian dynamics (BD) 
paths in a desired part of the domain, whilst using a continuum PDE-based 
model in the remainder. This PDE-assisted Brownian dynamics (PBD) methodology
has the advantage that efficient methods for solving PDEs can be used for 
large parts of the modelled domain, whilst BD data is available in other 
areas where required. The main goal of the PBD methodology is to get the same 
statistics (means and variances) in the BD subdomain as we would get
if we we were able to use BD simulations in the whole domain. 
In particular, the correct coupling between the 
two parts of the domain is of vital importance for the accuracy of a 
PBD algorithm. 

The paper is organised as follows. Section~\ref{sec:formulation} states, 
in mathematical terms, the requirements for the developed algorithms 
and introduces the notation used throughout the remainder of the paper. 
We then introduce the first PBD algorithm for a pure diffusion system in 
Section~\ref{sec:diffusion}, where we also explore the complications 
of this algorithm. In Section~\ref{sec:probation} we present the
second PBD algorithm which provides more accurate computations. 
In Section~\ref{sec:reactions} we investigate issues relating 
to the introduction of reactions into the system and present several 
computational examples. 

\section{Problem formulation} \label{sec:formulation}

Consider a general Brownian dynamics reaction-diffusion simulation 
with $M$ chemical species in the (open) domain $\Omega \subset {\mathbb R}^3$. 
We denote by $n_j(\mathbf{x}, t)$, $j=1,2,\dots,M,$ 
the expected spatio-temporal concentration of the $j$-th chemical 
species at the 
position $\mathbf{x}$ and time $t$ over our domain $\Omega$. 
The approximate mean-field reaction-diffusion PDEs for the 
time evolution of concentrations can be written as follows
\begin{equation}\label{eq:reactdiff}
 \frac{\partial p_j}{\partial t} = D_j\, \Delta p_j + R_j(p_1,p_2, 
 \dots,p_M)\,, \qquad j=1,2,\dots,M,
\end{equation}
where $p_j \equiv p_j(\mathbf{x}, t) : \Omega \times [0,\infty) 
\to [0,\infty)$ is the mean-field approximation of $n_j$,
$D_j$ is the diffusion constant of the $j$-th chemical species
and $R_j : [0,\infty)^M \to {\mathbb R}$ represents the reaction terms. 

The goal of PBD algorithms is to couple macroscopic description 
(\ref{eq:reactdiff}) in the open subdomain $\Omega_P\subset\Omega$ with 
a stochastic BD simulation in the open subdomain $\Omega_B\subset\Omega$,
where the closures of $\Omega_B$ and $\Omega_P$ cover $\Omega$,
i.e.
\begin{equation}
\Omega \subset \overline{\Omega_B} \cup \overline{\Omega_P}.
\label{omegacoverage}
\end{equation}
In $\Omega_B$, we will consider BD trajectories of individual molecules,
i.e. the state of the microscopic subdomain 
$\Omega_B$ is defined by the number $N_B^{(j)}(t)$ of molecules of 
the $j$-th chemical species at time $t$ and their positions 
${\mathbf x}_i^{(j)}(t) \in \Omega_B$, $i=1,2,\dots, N_B^{(j)}(t)$, 
$j=1,2,\dots, M$. We denote by $I$ the interface between the subdomains
$\Omega_B$ and $\Omega_P$, namely
\begin{equation}
I = \overline{\Omega_B} \cap \overline{\Omega_P}.
\label{interfaceI}
\end{equation}
In this paper, we will investigate two cases: 

\medskip 

{\leftskip 1.2cm

\parindent -6.5mm
{\bf [A]} $\Omega_B$ and $\Omega_P$ do not overlap, i.e. 
          $\Omega_B \cap \Omega_P = \emptyset$;

\smallskip

{\bf [B]} there exists an overlap region where
          the PDE description and BD simulations exist in parallel, i.e.
          $\Omega_B \cap \Omega_P \not = \emptyset.$

\medskip
\par}

\noindent
The case [A] will lead to the PBD algorithm (A1)--(A5) presented
in Table \ref{alg:complete}. The case [B] is implemented
in the second PBD algorithm (B1)--(B5) which is presented in
Table \ref{alg:probation}. We will start our discussion with
the case [A] because it is less technical to implement
than the case [B].   

To simplify our presentation, we will consider that $\Omega$
is a ``narrow" three-dimen\-sional domain, and hence only consider 
the process mapped onto an effective one-dimensional domain 
$\Omega\subset\mathbb{R}$ by assuming that the system is well 
mixed in the other two dimensions. In particular, we have
$\Omega_B \subset \Omega\subset\mathbb{R}$ and the
state of the BD subdomain $\Omega_B$ will be described by
the $x$-coordinates of molecules which we will denote as
$x_i^{(j)}(t) \in \Omega_B$, $i=1,2,\dots, N_B^{(j)}(t)$,
$j=1,2,\dots, M$.
We simulate the system using finite time step 
$\Delta t > 0$, in which particles in $\Omega_B$ change their 
position according to the discretized version of the overdamped 
Langevin equation
\begin{equation}
\label{eq:BD}
x_i^{(j)}(t+\Delta t) = x_i^{(j)}(t) + \sqrt{2D_j\Delta t} \, \xi\,,
\end{equation}
where $\xi$ is a normally distributed random variable with 
zero mean and unit variance. Figure~\ref{fig:sketch} shows 
a sketch of the described system for one chemical species along with 
the notation used in the case [A].
    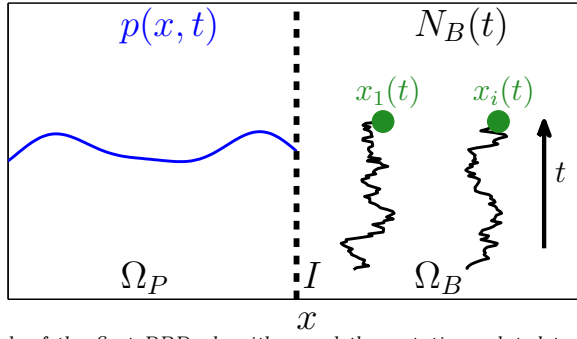
\begin{figure}
        \centering
        \begin{overpic}[width=0.6\textwidth]{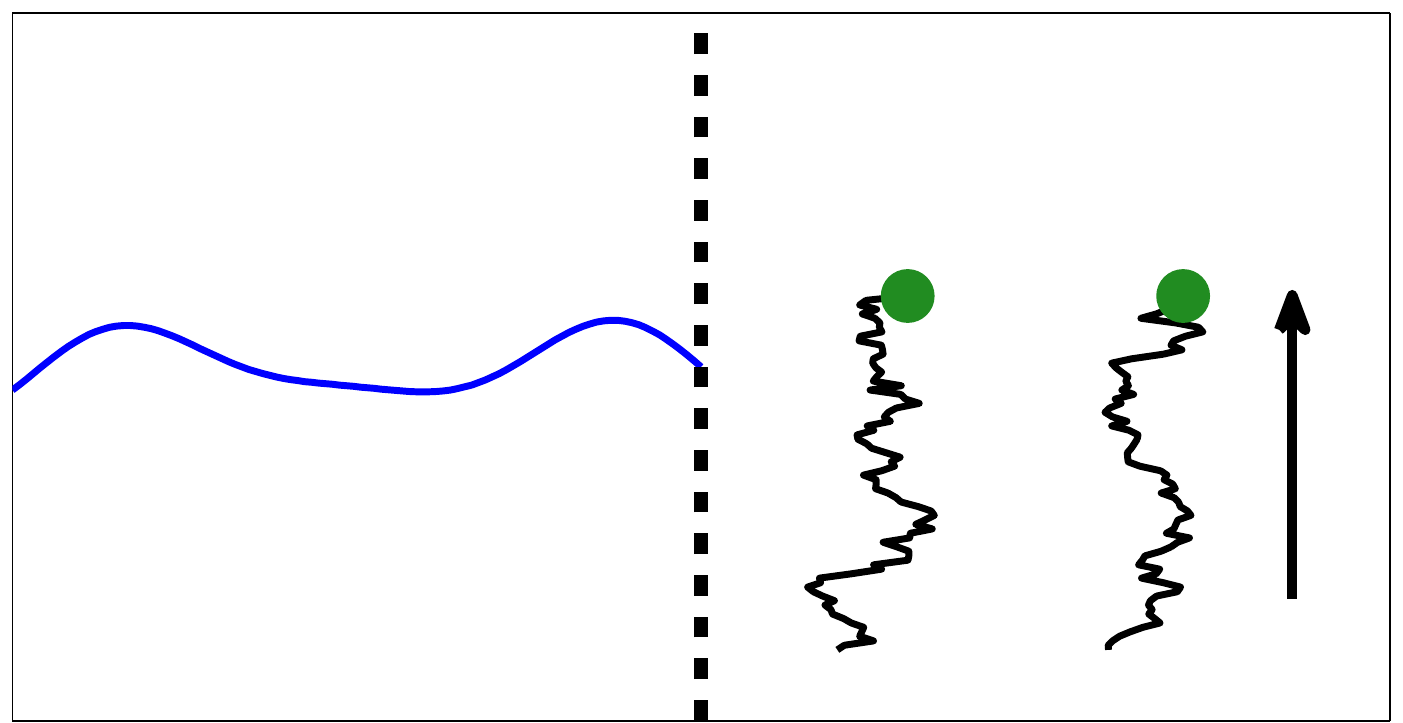}
            \put(20, 2){\Large $\Omega_P$}
            \put(70, 2){\Large $\Omega_B$}
            \put(51, 2){\Large $I$}
            \put(20, 45){\color{blue}\Large $p(x, t)$}
            \put(70, 45){\Large $N_B(t)$}
            \put(60, 34){\color{forestgreen}\large $x_1(t)$}
            \put(80, 34){\color{forestgreen}\large $x_i(t)$}
            \put(94, 20){\large $t$}
            \put(50, -5){\Large $x$}
        \end{overpic}
\caption{Sketch of the first PBD algorithm and the notation related to it. 
In $\Omega_P$, molecules are described by their density distribution $p(x,t)$ 
and in the microscopic domain $\Omega_B$ described by the number $N_B(t)$ 
of molecules and their positions $x_i(t)$, $i=1,2,\dots,N_B(t)$. 
The interface between these domains is denoted $I$.}
\label{fig:sketch}
\end{figure}
Our goal is to construct PBD algorithms which will satisfy
the following two conditions:

\medskip

\noindent
{\bf Condition (C.1):} We require that the expected distribution 
of molecules in $\Omega_B \setminus \Omega_P$ match that of the 
expected distribution $n_j(x,t)$, i.e. the distribution which 
we would obtain if we used detailed BD simulations in the whole 
domain $\Omega$. In particular this means that for every set 
$A\subset\Omega_B \setminus \Omega_P$, 
the expected number of particles in $A$ at time $t>0$ has to satisfy
\begin{equation}
\label{eq:problem:mean:part}
\mathbb{E}\left[\left|
\left\{ x_i^{(j)}(t)\in A\;, i=1,\ldots,N_B^{(j)}(t)\right\}\right|
\right] = \int_A n_j(x, t) \dx\,,\qquad \forall A\subset
\Omega_B\setminus \Omega_P\,.
\end{equation}
The choice of an arbitrarily small but finite interval
$A=[x, x+\dx)$ for $x\in \Omega_B\setminus \Omega_P$ leads to an alternative 
formulation of this condition
\begin{equation}
\label{eq:problem:mean:alt}
\mathbb{E}\left[\left|
\left\{ x_i^{(j)}(t)\in [x, x+\dx)\;, i=1,\ldots,N_B^{(j)}(t)\right\}\right|
\right] = n_j(x, t) \dx\,,\qquad \forall x \in \Omega_B\setminus \Omega_P\,.
\end{equation}

\medskip

\noindent
{\bf Condition (C.2):} 
Whilst we aim to match the expected outcome of the stochastic 
simulation to $n(x, t)$, we also want the variances of 
molecule distribution in $\Omega_B\setminus \Omega_P$ to match that which 
would be expected if a BD simulation were to be performed 
over the entire domain $\Omega$.

\medskip

\noindent
In $\Omega_P\setminus \Omega_B$, the system is described by
the concentration vector $p_j$, $j=1,2,\dots,M$, which evolves
according to the PDE (\ref{eq:reactdiff}). This distribution, 
whilst continuous, is not strictly deterministic since it is coupled 
with the stochastic outcomes of the BD subdomain $\Omega_B$.
If the stochastic reaction-diffusion model only includes zero-order
or first-order reactions, then the mean-field PDE (\ref{eq:reactdiff})
describes the expected behaviour of stochastic models
\cite{Erban:2007:PGS}. In this case it is reasonable to require
the following additional condition.

\medskip

\noindent
{\bf Condition (C.3):} We require
\begin{equation}
\label{eq:problem:mean:cont}
\mathbb{E}\left[ p_j(x, t)
\right] = n_j(x, t)\,,\qquad 
\forall x \in \Omega_P\setminus \Omega_B\,,\ t>0\,.
\end{equation}

\medskip

\noindent
In the case [A], the conditions (C.1)--(C.3) can be simplified by
observing that $\Omega_B = \Omega_B \setminus \Omega_P$
and $\Omega_P = \Omega_P\setminus \Omega_B$.
In the case [B], we will also require that the PBD algorithm 
gives the correct mean distribution of molecules in 
the overlap region $O = \Omega_B \cap \Omega_P$, i.e.
\begin{equation}
\label{eq:problem:mean:cont:overlap}
\mathbb{E}\left[ p_j(x, t) \right] \dx
+
\mathbb{E}\left[\left|
\left\{ x_i^{(j)}(t)\in [x, x+\dx)\;, i=1,\ldots,N_B^{(j)}(t)\right\}\right|
\right]
= n_j(x, t) \dx
\end{equation}
for all $x \in O = \Omega_B \cap \Omega_P$ and $t > 0.$
 
\section{PBD simulation of diffusion} \label{sec:diffusion}
In this section, we explain our case [A] PBD algorithm 
(i.e. $\Omega_B \cap \Omega_P = \emptyset$) using 
a system of diffusing non-interacting
molecules of a single chemical species. Therefore, for all further discussions
in this and the following section we will drop the index $j$ 
representing the species. Then the macroscopic PDE 
(\ref{eq:reactdiff}) for this system becomes the diffusion
equation
\begin{equation}\label{eq:diffusion}
 \frac{\partial n}{\partial t} = D \frac{\partial^2 n}{\partial x^2}\,,
\end{equation}
where $n \equiv n(x,t): \Omega \times [0,\infty) \to [0,\infty)$
and $D$ is the diffusion constant.
We will consider the infinite domain $\Omega = \mathbb{R}$ for
simplicity. Without loss of generality, we assume that 
$\Omega_P = (-\infty, 0)$ and $\Omega_B = (0, \infty)$, i.e.
the internal boundary $I = \{0\}$ is situated at the origin $x=0$. 
In the case of diffusion only, the total number of molecules $N$ 
in the system is conserved, i.e.
\begin{equation}
N = \int_{\Omega} n(x,t) \,\dx, \qquad \mbox{for all} \; t\geq 0\,.
\label{conservationofmass}
\end{equation}
Hence, for the PBD algorithm the conservation of mass condition
takes the form
\begin{equation}\label{eq:consmass}
N = \int_{\Omega_P} p(x, t) \,\dx + N_B(t)\,, 
\qquad \mbox{for all} \; t\geq 0\,.
\end{equation}
Since the diffusing molecules are non-interacting, we can express 
Condition (C.2) in mathematical terms. If all particles start with the
same initial condition, then each particle has
the (identical) probability $p_A=\int_A n(x,t)\dx/N$ of being 
in a set $A$ at time $t$. Consequently, the expected number of 
particles in $A$ at time $t$ is $N p_A$ and the variance is
equal to $Np_A(1-p_A)$. Substituting $p_A = \int_A n(x,t) \dx/N$ 
and using (\ref{conservationofmass}), we get Condition 
(C.2) in the following form
\begin{equation}
\label{eq:problem:var}
\operatorname{var}\left[\left|
\left\{ x_i(t)\in A\;, i=1,\ldots,N_B(t)\right\}\right|
\right] = \int_A n(x,t) \dx \left(1 - \frac{\int_A n(x, t)\dx}{
\int_\Omega n(x,t)\dx}\right)\,,
\end{equation}
for all $A\subset\Omega_B$ and $t>0$.
Using again $A=[x, x+\dx)$, we obtain the alternative formulation
\begin{equation}
\label{eq:problem:var:alt}
\operatorname{var}\left[\left|
\left\{ x_i(t)\in [x, x + \dx)\;, i=1,\ldots,N_B(t)\right\}
\right| \right] =  n(x, t) \dx\,.
\end{equation} 

In Sections~\ref{subsec:update:cont} and \ref{subsec:update:BD},
we present one update step of the continuum and the particle-based 
simulations respectively, before the full PBD algorithm (A1)--(A5)
is formulated 
in Section~\ref{subsec:fullalgorithm}. In Section~\ref{subsec:problems}, 
we will discuss the accuracy of the algorithm with respect to 
the Conditions (C.1)--(C.3).

\subsection{Updating the PDE regime in $\Omega_P$} 
\label{subsec:update:cont}
At time $t$, we have the concentration $p(x,t)$ for 
$x \in \Omega_P$ and are aiming to calculate the concentration
$p(x,t+\Delta t)$ that corresponds to a realisation of one 
time step $\Delta t$ of the diffusion process \eqref{eq:diffusion}.
We therefore define the exact outcome of a diffusion step in 
the full domain $\Omega$ given initial data $p(x,t)$:
\begin{equation} \label{eq:def:ptilde}
\widetilde{p}(x, t+\Delta t) = \int_{\Omega_P} 
K(x-x', \Delta t) \, p(x',t) \, \dx'\,,
\end{equation}
where $K(x-x',\Delta t)$ is the diffusion kernel
\begin{equation}\label{eq:def:kernel}
K(\xi, \Delta t) = \frac{1}{\sqrt{4\pi D \Delta t}} 
\exp\left(- \frac{\xi^2}{4D \Delta t}\right)\,,
\end{equation}
and the function $\widetilde{p}(x,t)$ has support $\Omega$.
Using this process, a certain proportion of the concentration 
distribution, namely
\begin{equation} \label{eq:def:alpha}
\alpha(t + \Delta t)
\equiv 
\int_{\Omega_B} \widetilde{p}(x, t + \Delta t) \, \dx
\end{equation}
would have crossed the interface $I$ in the time interval 
$[t, t + \Delta t)$ 
(see Figure~\ref{fig:deltat}).
\begin{figure}
        \centering
        \begin{overpic}[width=0.6\textwidth]{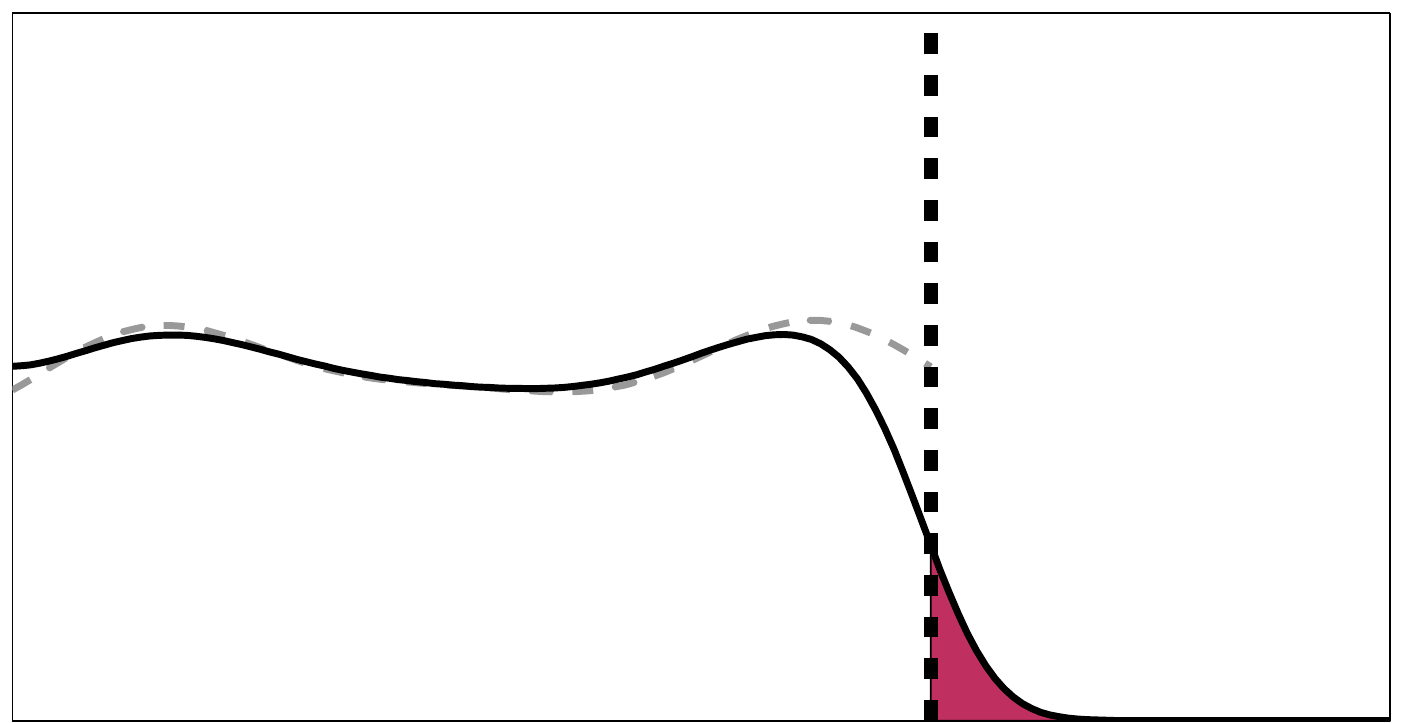}
            \put(30, 2){\Large $\Omega_P$}
            \put(85, 2){\Large $\Omega_B$}
            \put(62, 2){\Large $I$}
            \put(25, 45){\color{grey}\Large $p(x, t)$}
            \put(25, 35){\Large $\widetilde{p}(x, t+\Delta t)$}
            \put(70, 10){\color{maroon}\Large $\alpha(t+\Delta t)$}
            \put(50, -5){\Large $x$}
        \end{overpic}

\caption{Illustration of step {\rm (A1)} of the first PBD algorithm. 
Dashed line: $p(x, t)$; solid line: $\widetilde{p}(x, t+\Delta t)$ 
as defined in \eqref{eq:def:ptilde}; shaded area: $\alpha(t+\Delta t)$ 
as defined in \eqref{eq:def:alpha}.}
\label{fig:deltat}
\end{figure}
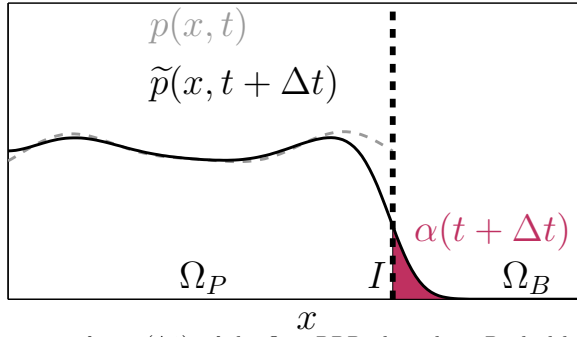
This value $\alpha(t+\Delta t)$ represents the expected number of 
molecules to cross the interface from $\Omega_P$ to $\Omega_B$ in 
the time interval $[t, t+\Delta t)$. Since all molecules in $\Omega_P$ 
are identical (of the same chemical species and have the same probability 
distribution proportional to $\widetilde{p}(x,t+\Delta t)$), 
the number of molecules that cross the interface 
from $\Omega_P$ to $\Omega_B$ in the time interval $[t, t+\Delta t)$ 
is Poisson distributed with  average  $\alpha(t + \Delta t)$. 
Let us assume that $\Delta t$ has been chosen small enough to ensure 
that $\alpha(t + \Delta t) \ll 1$. In this case 
the probability of more than one particle crossing the interface
is negligible and we need to consider two cases:

{
\leftskip 1cm

\smallskip

\noindent \textbf{(i)} one particle gets created in $\Omega_B$ with 
the probability $\alpha(t + \Delta t)$;

\leftskip 0.95cm

\noindent \textbf{(ii)} no particle is created with the probability 
$1 - \alpha(t+\Delta t)$.

\smallskip
\par}

\noindent For both cases we need to calculate the updated 
concentration $p(x, t+\Delta t)$ where $x\in\Omega_P$.

We consider the concentration $p(x,t)$ as the distribution of 
$N-N_B(t)$ identically distributed particles at time $t$. Therefore 
each of these particles has at time $t$ the probability distribution 
$p(x, t)/(N - N_B(t))$. After one time step each of these particles 
can be found in an infinitesimal interval $[x, x+\dx)$ with  probability
$\widetilde{p}(x,t+\Delta t) \, \dx$. For each particle, its probability 
distribution given that it did not leave $\Omega_P$ can be calculated as
\begin{equation*}
 p_1(x, t+\Delta t) = \frac{\widetilde{p}(x, t+\Delta t)}{N - N_B(t) - \alpha(t+\Delta t)}\,, 
 \qquad \mbox{for} \; x\in\Omega_P\,.
\end{equation*}
On the other hand, if a particle does leave the domain $\Omega_P$, 
then its distribution function becomes zero for
$x\in\Omega_P$. Instead the particle  is introduced 
into $\Omega_B$ at time $t+\Delta t$ at a location $x$ with the 
probability distribution given by
\begin{equation} \label{eq:p2}
p_2(x, t+\Delta t) = \frac{\widetilde{p}(x, t+\Delta t)}{\alpha(t+\Delta t)}\,,
\qquad \mbox{for} \; x\in\Omega_B\,.
\end{equation}
This updating process is rather like collapsing a wavefunction in
quantum mechanics \cite{vonNeumann:1955:MFQ}. 
We have a look to see if a particle has crossed the
boundary into $\Omega_B$: if it is there on the other side then 
its distribution function collapses to a $\delta$ function at its new 
position, while if it is not, then the distribution function collapses 
to zero in $\Omega_B$ with corresponding rescaling in $\Omega_P$.

Combining these arguments for all the particles, we see that the last
update step is a simple rescaling of  
the probability distribution so that the updated distribution
satisfies conservation of mass  
according to \eqref{eq:consmass}. We therefore have
\begin{equation}
p_{\mathrm{cont}}(x,t+\Delta t) 
= \left\{\begin{array}{ll}
   \beta_{(i)} \, \widetilde{p}(x, t+\Delta t), & \mbox{in the case (i),} 
\\
\beta_{(ii)} \, \widetilde{p}(x, t+\Delta t), & \mbox{in the case (ii).}
\end{array}\right. \quad \mbox{for} \; x\in\Omega_P\,,
\label{pcontdistr}
\end{equation}
with the rescaling constants $\beta_{(i)},$ $\beta_{(ii)}$ given by
\begin{equation} \label{eq:defn:betas}
\beta_{(i)} = \frac{N - N_B(t) -1}{N - N_B(t) - \alpha(t + \Delta t)}\,,
\qquad
\beta_{(ii)} = \frac{N - N_B(t)}{N - N_B(t) - \alpha(t + \Delta t)}\,.
\end{equation}
Note that the update step for the continuum regime satisfies 
conservation of mass \eqref{eq:consmass}. An 
illustration of the two cases can be seen in Figure~\ref{fig:agentcreated}.
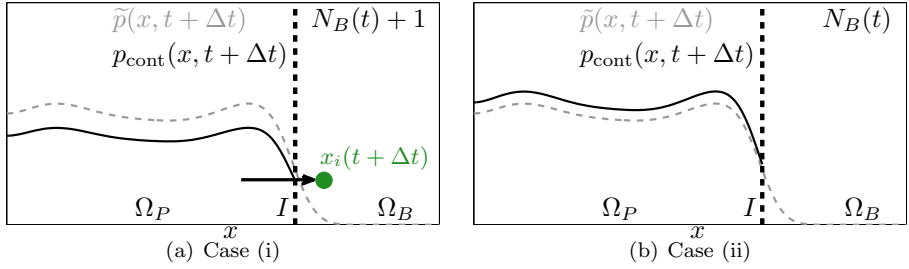
\begin{figure}
        \centering
        \subfigure[Case (i)]{
        \begin{overpic}[width=0.45\textwidth]{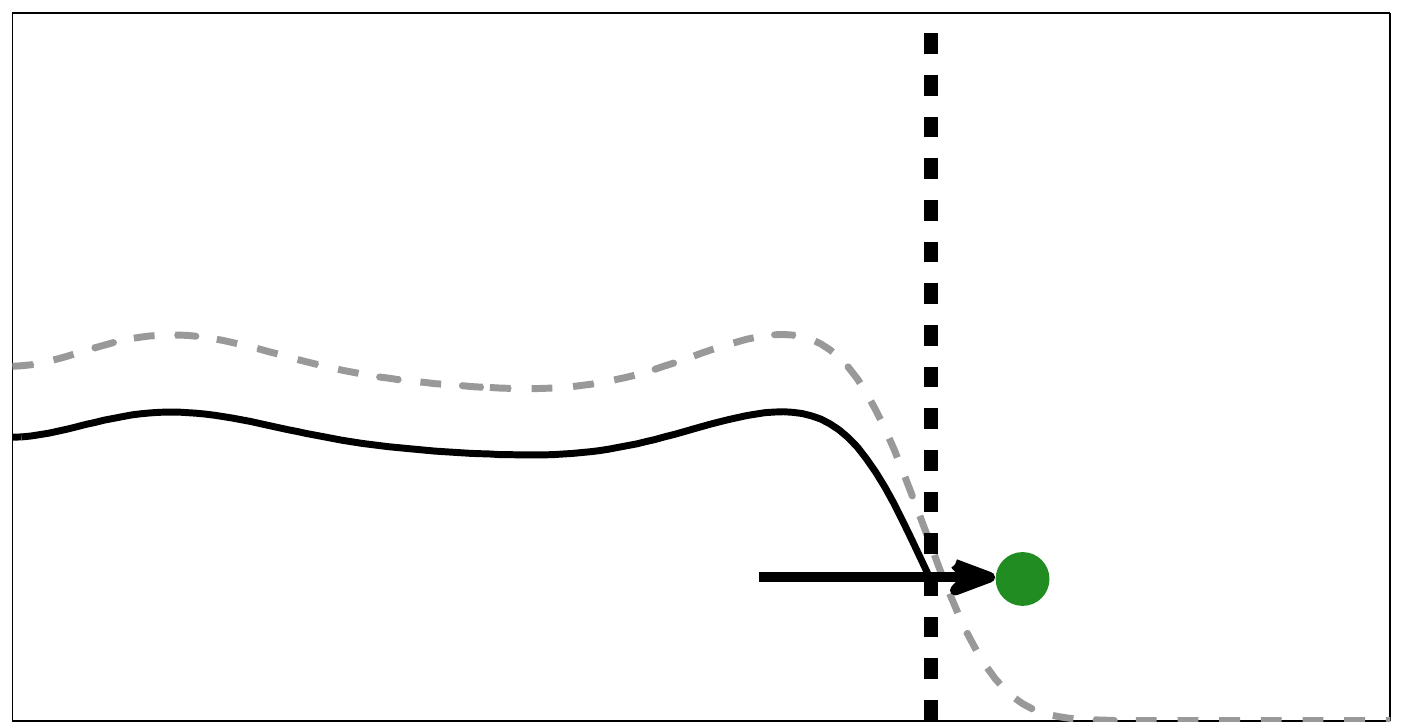}
            \put(30, 2){$\Omega_P$}
            \put(85, 2){$\Omega_B$}
            \put(62, 2){$I$}
            \put(50, -3){$x$}
            \put(25, 45){\color{grey} $\widetilde{p}(x, t+\Delta t)$}
            \put(25, 37){$p_{\mathrm{cont}}(x, t+\Delta t)$}
            \put(72, 14){\small\color{forestgreen} $x_i(t+\Delta t)$}
            \put(70, 45){$N_B(t) + 1$}
        \end{overpic}
        }
        \subfigure[Case (ii)]{
        \begin{overpic}[width=0.45\textwidth]{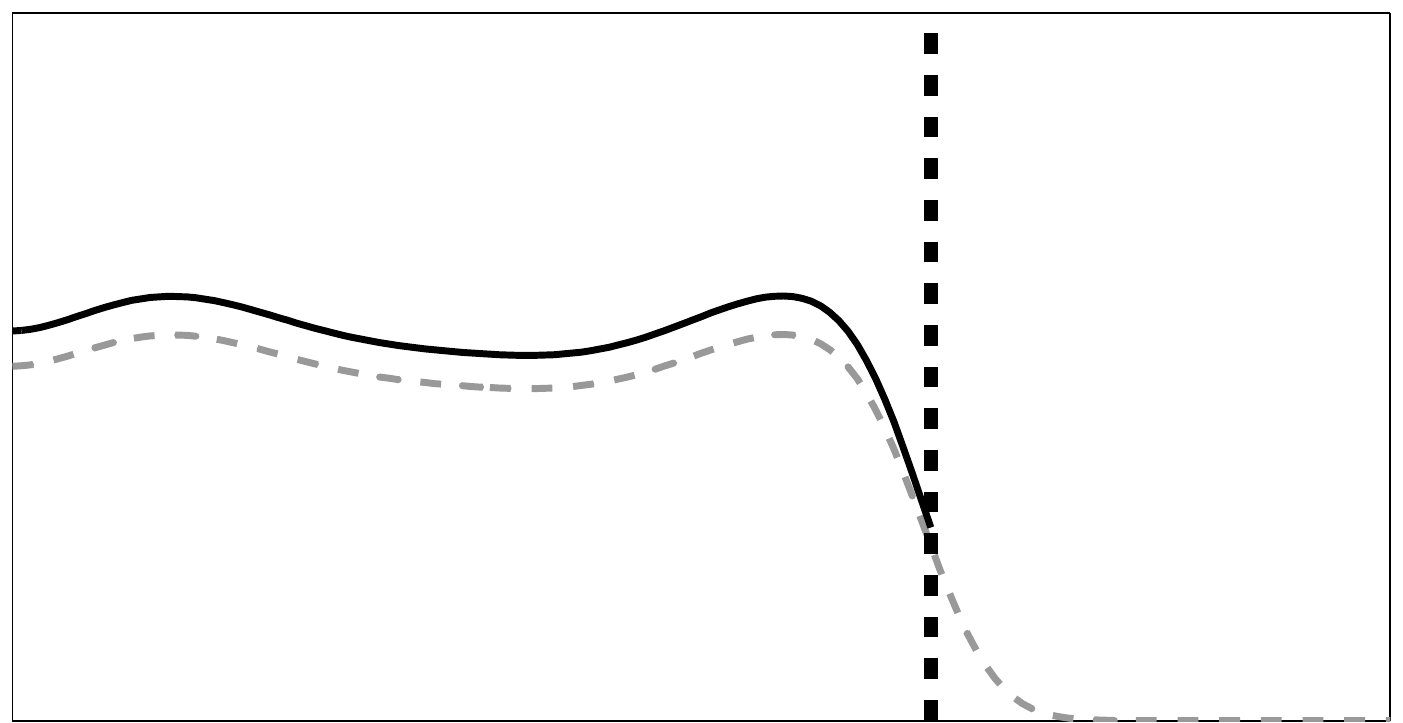}
            \put(30, 2){$\Omega_P$}
            \put(85, 2){$\Omega_B$}
            \put(62, 2){$I$}
            \put(50, -3){$x$}
            \put(25, 45){\color{grey} $\tilde{p}(x, t+\Delta t)$}
            \put(25, 37){$p_{\mathrm{cont}}(x, t+\Delta t)$}
            \put(80, 45){$N_B(t)$}
        \end{overpic}
        }
\caption{Illustration of steps {\rm (A2)} and
{\rm (A3)} of the first PBD algorithm.  
Dashed line: $\widetilde{p}(x, t+\Delta t)$; 
solid line: $p_{\mathrm{cont}}(x, t+\Delta t)$; 
(green) circle: created particle.}
\label{fig:agentcreated}
\end{figure}

\subsection{Updating the BD regime in $\Omega_B$} \label{subsec:update:BD}
We use the discretized version of the overdamped Langevin equation 
introduced in \eqref{eq:BD} to update the positions
of particles in $\Omega_B$. If the position of the $i$-th molecule, 
computed by \eqref{eq:BD}, is in $\Omega_B$ at the end of the time step,
then we will continue representing it as a particle. Note, 
that a particle that crossed the boundary $I$ and came back into 
$\Omega_B$ during the time step $[t, t + \Delta t)$ is also captured 
by this case. We have to be more careful whenever the position 
$x_i(t+\Delta t)$, computed by \eqref{eq:BD}, is inside the PDE subdomain 
$\Omega_P$ at time $t+\Delta t$, i.e. $x_i(t+\Delta t) \in \Omega_P$. 
In this case, the random walk of the $i$-th molecule crossed the interface 
$I$ without crossing back at the end of the time step. 
This event needs to be taken into account for 
the update of the final concentration $p(x, t+\Delta t)$ in $\Omega_P$. 
As we know the exact position of this particle $i$ at time $t+\Delta t$ 
we add a Dirac $\delta$ function at the position 
$x_i(t+\Delta t)\in\Omega_P$. 
Therefore we compute $p(x,t+\Delta t)$ in $\Omega_P$ by
\begin{equation*}
p(x,t+\Delta t) = p_{\mathrm{cont}}(x,t+\Delta t) 
+ 
\sum_{x_i(t+\Delta t) \in \Omega_P} \delta(x - x_i(t + \Delta t)),
\end{equation*}
where $p_{\mathrm{cont}}(x,t+\Delta t)$ is given
by (\ref{pcontdistr}).

\subsection{The first PBD algorithm} \label{subsec:fullalgorithm}
This algorithm computes the concentration $p(x, t)$ for  $x\in\Omega_P$,
and the number $N_B(t)$ and positions of BD particles 
$x_i(t)\in\Omega_B$, $i=1,2,\dots,N_B(t).$
One time step of the first PBD algorithm is presented in 
Table~\ref{alg:complete} as algorithm (A1)--(A5). 
In order to simplify the presentation
of this algorithm, we consider that the time step $\Delta t$ is
chosen so small that $\alpha(t + \Delta t) \ll 1$. In particular,
we only need to implement cases (i)--(ii) presented in
Section \ref{subsec:update:cont}, because the probability
that two or more molecules are initiated in $\Omega_B$ 
during one time step is negligible.

The auxiliary distribution $\widetilde{p}(x, t+\Delta t)$ in step (A1) 
is in practice calculated using a numerical approximation algorithm. 
To calculate $\alpha(t+\Delta t)$ we can either use this numerical 
approximation of $\,\widetilde{p}(x, t+\Delta t)$, which requires 
an additional time step $\widetilde{\Delta t} \ll \Delta t$ to be 
used to ensure 
the accuracy of $\alpha(t+\Delta t)$, or (more efficiently) we can approximate
$\alpha(t+\Delta t)$ analytically using a boundary layer expansion in
the vicinity of the interface $I$.

\renewcommand{\labelenumi}{(A\arabic{enumi})}
\renewcommand{\labelenumii}{\textbf{(\roman{enumii})}}
    \begin{table}[t]
    \begin{framed}
    \begin{enumerate}
    \setcounter{enumi}{0}
        \item \label{algorithmpos:2}
                Calculate $\widetilde{p}(x, t+\Delta t)$ using 
		\eqref{eq:def:ptilde} and $\alpha(t+\Delta t)$
                using \eqref{eq:def:alpha}.
        \item Generate uniformly distributed random number $r$ in $(0,1)$. 
        \begin{enumerate}
           \item If $r < \alpha(t+\Delta t)$, then create new particle 
	   in $\Omega_B$ according to the probability density
           $p_2(x, t+\Delta t)$ defined in \eqref{eq:p2}. Set
	   $\beta = \beta_{(i)}$ where $\beta_{(i)}$ is given
	   by (\ref{eq:defn:betas}). Set $N_{B,1} = 1.$
           \item If $r \geq \alpha(t+\Delta t)$, then set
	   $\beta = \beta_{(ii)}$ where $\beta_{(ii)}$ is given
	   by (\ref{eq:defn:betas}). \\ Set $N_{B,1} = 0.$
	\end{enumerate}

        \item Compute positions $x_i(t+\Delta t)$, $i=1,2,\dots,N_B(t)$, 
	      of BD particles according to \eqref{eq:BD}. 
        \item Compute new concentration in $\Omega_P$ by
            \begin{equation*}
                p(x, t+\Delta t) = \beta \,
		 \widetilde{p}(x, t+\Delta t) 
		 + \sum_{x_i(t+\Delta t) \in \Omega_P} \!\!\!\!\!
		 \delta(x - x_i(t + \Delta t))\,, \quad 
		 \mbox{for} \; x\in\Omega_P\,.
            \end{equation*}
         \item  Update the number of BD particles by  
	  \begin{equation*}
                N_B(t+\Delta t) = N_B(t) 
		+ 
		N_{B,1} 
		- 
		\left|\left\{ x_i(t+\Delta t)\in\Omega_P\;, i=1,\dots,N_B(t)
		\right\}\right|
                \,.
            \end{equation*}
	    Terminate computation of trajectories of BD molecules which 
	    landed in $\Omega_P$ (i.e. the BD particles which satisfy 
	    $x_i(t+\Delta t)\in\Omega_P.$) \\
	Then continue with step (A1) for time $t+\Delta t$.
\end{enumerate}
\end{framed}
\caption{One time step of the first PBD algorithm for 
a system of diffusing molecules.}
\label{alg:complete}
    \end{table}

By construction, the algorithm (A1)--(A5) satisfies the conservation of 
mass condition \eqref{eq:consmass} and the concentration $p(x,t)$ 
satisfies non-negativity. We will now show that this algorithm also 
guarantees the correct expected outcome and therefore satisfies 
Conditions (C.1) and (C.3).

\begin{theorem}
\label{thm:1}
Consider the BD simulation of $N$ diffusing molecules in the 
computational domain $\Omega$ which is divided into sudomains 
$\Omega_B \subset \Omega$ and $\Omega_P \subset \Omega$ satisfying 
$(\ref{omegacoverage})$
and the case {\rm [A]}. Suppose that $N_B(0)$ particles are initially 
in $\Omega_B$ at positions $x_i(0)$, $i=1,2,\dots,N_B(0).$ Let
us initialize $p(x,0)$ as sums of Dirac $\delta$ functions describing 
molecules which are initially in $\Omega_P$, i.e. $p(x,0) = n(x,0)$ 
for $x \in \Omega_P$. Then the expected outcome of the PBD algorithm 
{\rm (A1)--(A5)} (presented in Table~$\ref{alg:complete}$) 
satisfies the Conditions {\rm (C.1)} and {\rm (C.3)}
for arbitrary $\Delta t > 0$.
\end{theorem}

\begin{proof}
We will show that the identities \eqref{eq:problem:mean:part} and
\eqref{eq:problem:mean:cont} hold during one iteration
(A1)--(A5) presented in Table~\ref{alg:complete}. It will then follow by 
induction that they hold for all times $k\Delta t$, $k=0,1,2,\dots$. 

Let us assume that the probability distribution of particles in 
$\Omega_B$ at time $t$ is given through $n(x, t)$, $x\in\Omega_B.$
Given the probability distribution $p(x,t)$, $x\in\Omega_P$, at
time $t$, the conditional expected value of $p(x, t+\Delta t)$ 
for $x\in\Omega_P$ is given by
\begin{eqnarray*}
\mathbb{E}\,\big[\,p(x, t+\Delta t)\,|\,p(x,t)\,\big] 
            & = & \alpha(t+\Delta t) \, 
	    \beta_{(i)} \, \widetilde{p}(x,t+\Delta t) \\
            & + & (1-\alpha(t+\Delta t)) \, \beta_{(ii)} \, 
	    \widetilde{p}(x, t+\Delta t) \\
            & + & \int_{\Omega_B} K(x-x', \Delta t) \, n(x', t) \, \dx'\,,
\end{eqnarray*}
 where $\beta_{(i)}$ and $\beta_{(ii)}$ are given
 by (\ref{eq:defn:betas}),
 $K(x-x',\Delta t)$ is the diffusion kernel given in \eqref{eq:def:kernel}
 and the last term represents particles 
 that moved across the interface 
 $I$, defined by (\ref{interfaceI}), during the time step $[t, t+\Delta t)$.
 Using \eqref{eq:defn:betas}, we obtain
$$ 
\mathbb{E}\,\big[\,p(x, t+\Delta t)\,|\,p(x,t)\,\big] 
= \widetilde{p}(x, t+\Delta t) + 
            \int_{\Omega_B} K(x-x', \Delta t) \, n(x', t) \, \dx'.
$$   
Using \eqref{eq:def:ptilde} and the law of total expectation
(law of iterated expectations), we get
$$
\mathbb{E}\,\big[p(x, t+\Delta t)\,\big] 
      = \int_{\Omega_P} \! K(x-x', \Delta t) \, \mathbb{E}[p(x',t)] \, \dx'
	      +
	\int_{\Omega_B} \! K(x-x', \Delta t) \, n(x',t) \, \dx'.
$$    
Using the induction assumption that $\mathbb{E}[p(x, t)] = n(x, t)$, 
we obtain 
$$        
\mathbb{E}\,\big[p(x, t+\Delta t)\,\big] 
            = \int_{\Omega} n(x', t) K(x-x', \Delta t) \, \dx' 
	    = n(x, t+\Delta t) \,,\quad \mbox{for} \; x \in\Omega_P\, ,
$$
i.e. we have derived \eqref{eq:problem:mean:cont}. 

Let us consider a set $A\subset\Omega_B$. 
Given the probability distribution $p(x,t)$, $x\in\Omega_P$, at
time $t$, the conditional
expected number of particles in $A$ at time $t+\Delta t$ is
\begin{eqnarray*}
\mathbb{E}\,\big[\left|
\left\{ x_i(t)\in A\;, i=1,\ldots,N_B(t)\right\}\right|
\, \big| \, p(x,t)
\,\big] 
&=&
\int_A 
\alpha(t+\Delta t) \, p_2(x,t+\Delta t) 
\, \dx \\
&+&
\int_A \int_{\Omega_B} K(x-x', \Delta t) \, n(x',t) \, \dx' \, \dx,
\end{eqnarray*}
where the first term represents newly created particles from the 
PDE regime $\Omega_P$ and the second term represents the movement of 
particles inside $\Omega_B$. Using \eqref{eq:def:ptilde},
(\ref{eq:p2}) and the law of total expectation, we obtain
\begin{eqnarray*}
\mathbb{E}\,\big[\, \big|
\, \left\{ x_i(t)\in A\;, i=1,\ldots,N_B(t)\right\} \, \big| \,
\,\big]
 &=& \int_A
 \int_{\Omega} 
 K(x-x', \Delta t) \, n(x',t) \, \dx' \, \dx \\
 &=& \int_A n(x, t+\Delta t)\, \dx\,,
\end{eqnarray*}
which is the condition \eqref{eq:problem:mean:part}. Thus we have
showed that both Conditions (C.1) and (C.3) are satisfied. This concludes
the proof.
\end{proof}

\medskip

Theorem \ref{thm:1} also holds if the algorithm (A1)--(A5) is extended 
to the creation of more than one new particle per time step as long 
as the expected value of the number of created particles 
is $\alpha(t+\Delta t)$ and the rescaling is done accordingly.
The algorithm (A1)--(A5) and the proof can be easily extended 
for a finite domain $\Omega = [0, L]$ with no flux boundary 
conditions by redefining the kernel $K(\xi, \Delta t)$ accordingly.

\subsection{Discussion of the PBD algorithm (A1)--(A5)}
\label{subsec:problems}
In Theorem~\ref{thm:1} we showed that the PBD algorithm (A1)--(A5)
satisfies the Conditions (C.1) and (C.3). However, we still need 
to check whether the Condition (C.2) on the variances is also satisfied.

To investigate the variances created by this algorithm, we show 
the outcome of an illustrative numerical example in Figure~\ref{fig:example1}.
We simulate the  diffusion of 100 molecules in the domain $\Omega = [-1, 1]$
which are initialized at the same location $x=-0.95$.
We use no flux boundary conditions. We test the algorithm (A1)--(A5) 
where $\Omega_P = (-1,0)$, $\Omega_B = (0,1)$ and $I = \{0\}$.
To calculate $\widetilde{p}(x,\Delta t)$ we use an implicit 
Euler-scheme with $\Delta x = 0.01$ and a numerical time-step of 
$\widetilde{\Delta t} = 10^{-6}$. The time step $\Delta t$ used 
by the algorithm (A1)--(A5) is $\Delta t = 10^{-3}$ 
and we simulate the system until $T_{\mathrm{final}} = 0.2$. The
time step $\widetilde{\Delta t} \ll \Delta t$ for the approximation of 
$\widetilde{p}(x, t+\Delta t)$ was chosen in order to minimise numerical 
artefacts. We run 1000 realisations of this process and measured the number 
of particles in 10 intervals (`bins') of the size $0.1$ in $\Omega_B$.
Averaging over 1000 realisations, we calculate the mean value and the 
variance of the particle number for each of the bins at time
$T_{\mathrm{final}}$. The results are presented in 
Figure~\ref{fig:example1} as gray histograms. In this example, 
it is easy to calculate
the correct distribution $n(x,t)$ which the algorithm (A1)--(A5)
tries to approximate. It can be obtained by solving the diffusion 
equation \eqref{eq:diffusion} in the domain $\Omega = (-1, 1)$ with 
$n(x, 0) = 100 \, \delta(x+0.95)$ and no flux boundary conditions.
The expected values for both means and variances are plotted as 
(red) dashed lines in Figure~\ref{fig:example1}.  
    \begin{figure}
        \center
        \subfigure[Mean]{
            \epsfig{file=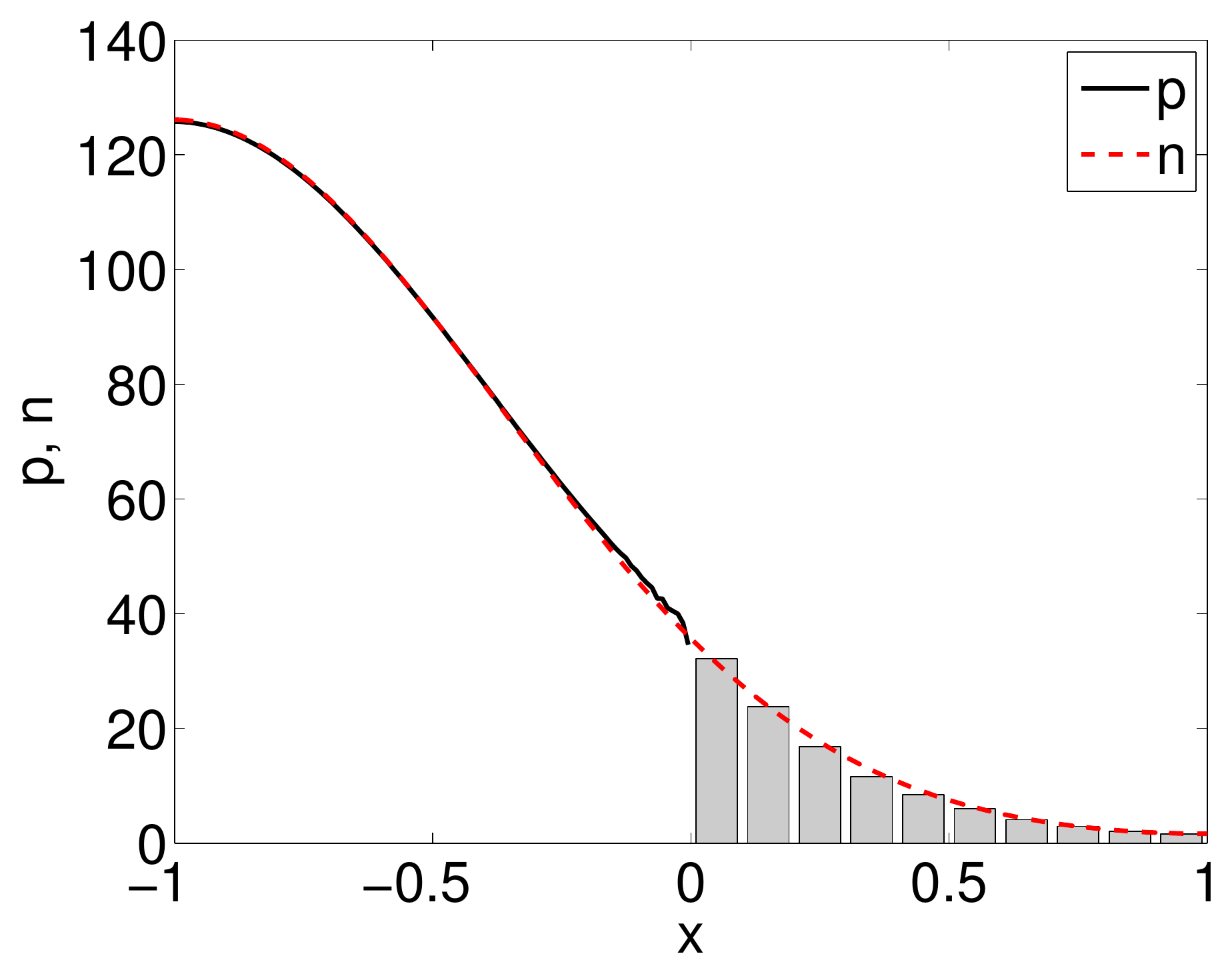, width=0.45\textwidth}
        }
        \subfigure[Variance]{
            \epsfig{file=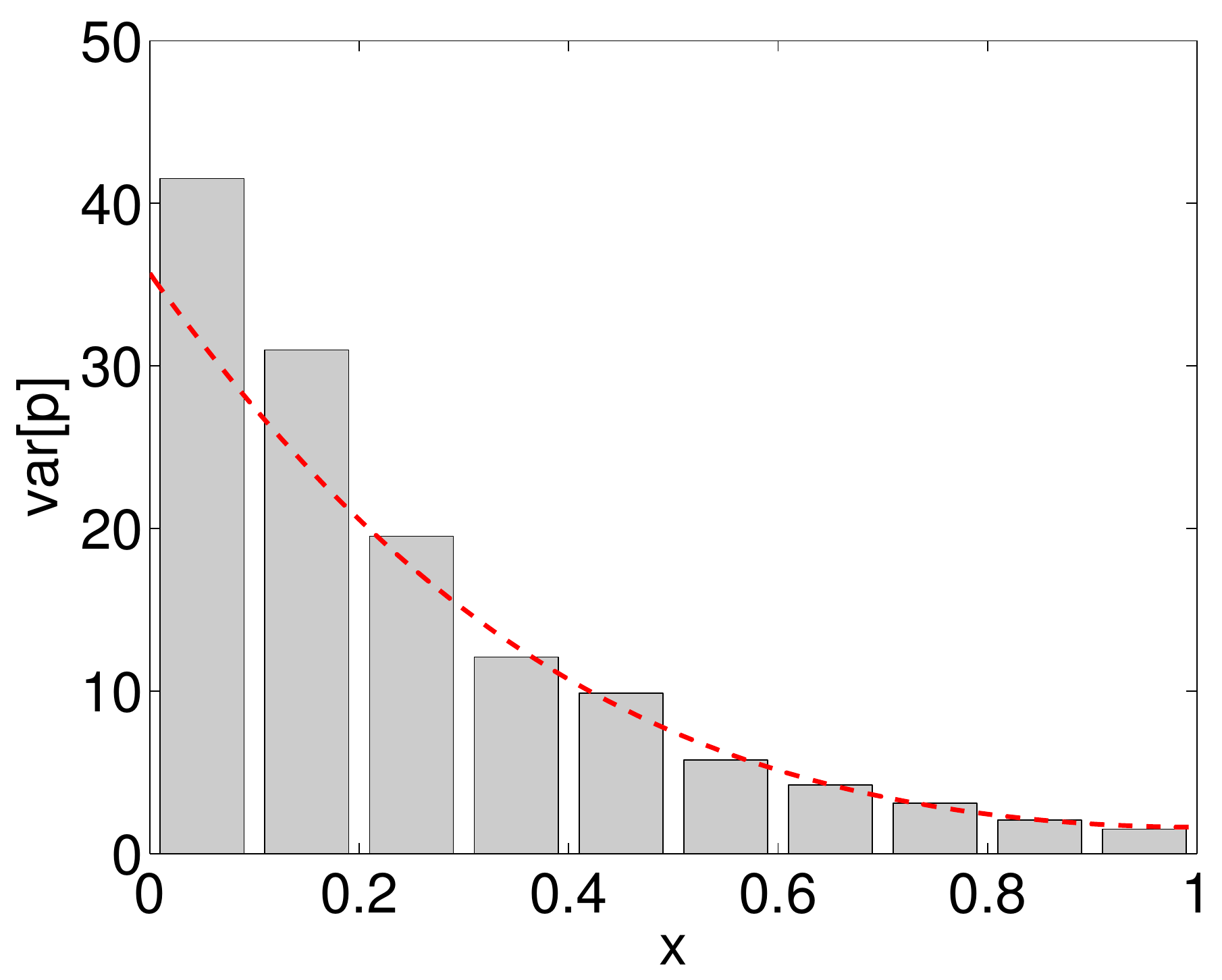, width=0.45\textwidth}
            \label{fig:example1:var}
        }
\caption{Simulation results of diffusion of 100 particles in $(-1,1)$ with 
no flux boundary conditions initialized at $x=-0.95,$
i.e. $n(x, 0) = 100 \, \delta (x+0.95)$. Results averaged over 
1000 realisations. Dashed (red) line: expected outcome. 
{\rm (a)} Solid line: mean value in $\Omega_P$; (gray) bars: 
particle concentrations in $\Omega_B$. {\rm (b)} 
(Gray) bars: measured variances in particle concentrations. 
Parameter values are described in the text.}
\label{fig:example1}
\end{figure}

In Figure \ref{fig:example1}(a), we see that the mean value of the 
simulation results matches well with the solution of the diffusion 
process, with only small fluctuations close to the internal boundary $I$ 
at $x=0$ due to stochastic effects and inaccuracies caused by 
the numerical approximation of $\widetilde{p}(x, t+\Delta t)$. 
However, in Figure~\ref{fig:example1:var} it becomes clear that 
the variance between different realisations is higher than the 
desired value, in particular close to the internal interface. 
We can explain this effect clearly using a thought experiment.

Let us consider a situation where $p(x,0)$ is $0$ close to 
the internal boundary $I$ and has a peak of mass 99 arbitrarily 
far away from $I$. Additionally, we assume that one particle is 
situated in $\Omega_B$ close to the interface $I$. Assuming the 
particle crosses the interface in the first simulation step, a 
Dirac $\delta$ function is created in $\Omega_P$ close to the interface, 
as illustrated in Figure~\ref{fig:thexp1}. This $\delta$ function 
has a large impact on the region in $\Omega_P$ that is close to the 
interface, as the distribution $p(x, \Delta t)$ is negligible in 
this region. Immediately after incorporating the particle into 
$p(x, \Delta t)$, all its information is lost and we are forced 
to assume $100$ independent particles with the probability 
distribution $p(x, \Delta t)$ in $\Omega_P$. In particular 
this implies that every particle has a $1\%$ chance of being 
at the position of the $\delta$ and $99\%$ chance of being 
in the bulk distribution far away from the boundary. This is 
indeed not the case, since we know that there is exactly one 
molecule near the interface and 99 molecules away from the 
interface, but the nature of the continuum distribution means 
that this information must be lost or else we should necessarily 
demand a separate distribution for all molecules in $\Omega_P$.

In the second diffusion step, we calculate $\widetilde{p}(x, 2\Delta t)$ 
according to \eqref{eq:def:ptilde} and some `mass' 
$\alpha(2\Delta t)$ may have drifted across the interface $I$ 
(see Figure~\ref{fig:thexp2}). Because the bulk distribution 
is far away from the boundary, almost all of this mass 
$\alpha(2\Delta t)$ comes from the $\delta$ function close to 
$I$. Let us imagine that a new particle is now created in 
$\Omega_B$ (according to the probability $\alpha(2\Delta t)$), 
in which case the whole distribution 
needs to be rescaled, as shown in Figure~\ref{fig:thexp3}. 
Because $99\%$ of the mass in $p(x, \Delta t)$ is situated 
in the bulk far away from the boundary, the majority of 
the rescaling happens in this region, such that effectively 
the mass needed to create the particle is almost entirely 
taken from the bulk, rather than from the region close to 
the interface. This also implies that most of the mass close 
to the boundary will stay and it is therefore possible to create 
another particle from this mass in further time steps. This 
is in contradiction to the result that would be expected if 
information was not lost in the first time step. That is, 
the distribution close to the interface should dissapear 
and the bulk far from the interface is left alone. This 
effect is the main reason a higher than expected variance 
can be measured in $\Omega_B$. Note, however, that this 
does not effect the expected values, as shown in Theorem~\ref{thm:1}.
     \begin{figure}
        \center
        \subfigure[Particle jumps into $\Omega_P$ and creates 
	a Dirac $\delta$ function at the landing position.]{
            \label{fig:thexp1}
            \begin{overpic}[width=0.45\textwidth]{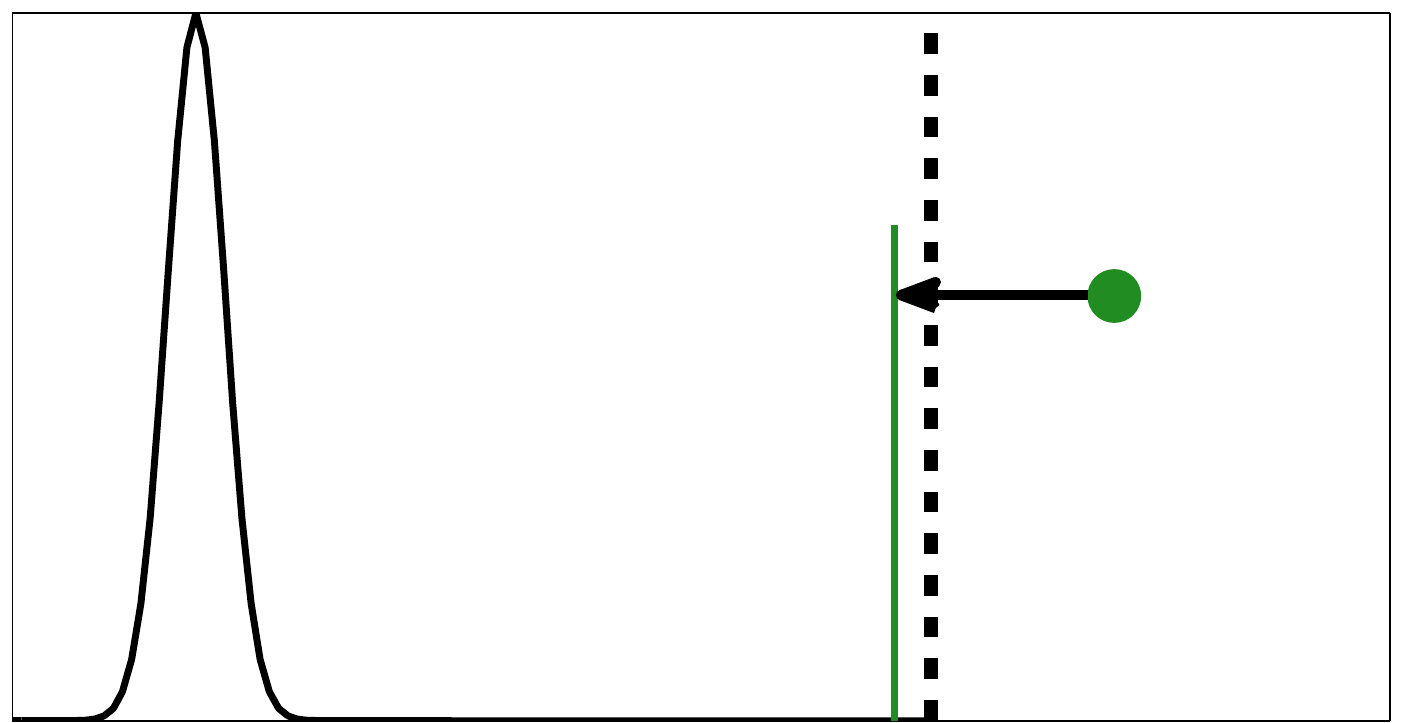}
                \put(30, 2){$\Omega_P$}
                \put(85, 2){$\Omega_B$}
                \put(67, 2){$I$}
                \put(50, -3){$x$}
                \put(83, 29){\color{forestgreen} $x_i(0)$}
                \put(30, 30){\color{forestgreen}\small $\delta(x-x_i(\Delta t))$}
                \put(25, 45){$p(x, \Delta t)$}
                \put(11, 2){99}
                \put(60, 2){1}
            \end{overpic}
        }
        \subfigure[The distribution diffuses for one time step and an amount of the $\delta$ flows back across the boundary.]{
            \label{fig:thexp2}
            \begin{overpic}[width=0.45\textwidth]{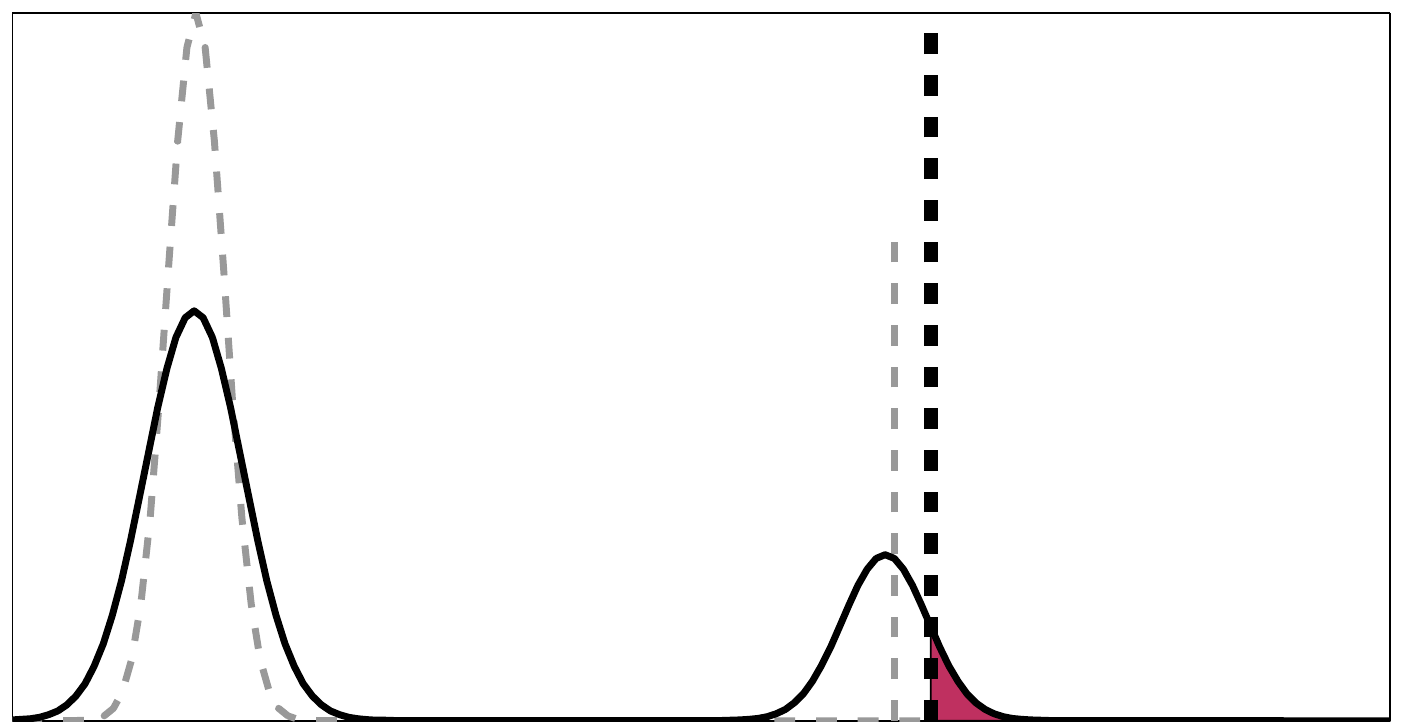}
                \put(30, 2){$\Omega_P$}
                \put(85, 2){$\Omega_B$}
                \put(67, 20){$I$}
                \put(68, 5){\color{maroon}\small$\alpha(2\Delta t)$}
                \put(50, -3){$x$}
                \put(25, 45){\color{grey} $p(x, \Delta t)$}
                \put(25, 37){$\widetilde{p}(x, 2\Delta t)$}
                \put(11, 2){99}
                \put(60, 2){1}
            \end{overpic}
        }
        \subfigure[If a particle is created, the majority of its mass is virtually taken from the bulk distribution and the peak near the boundary remains.]{
            \label{fig:thexp3}
            \begin{overpic}[width=0.45\textwidth]{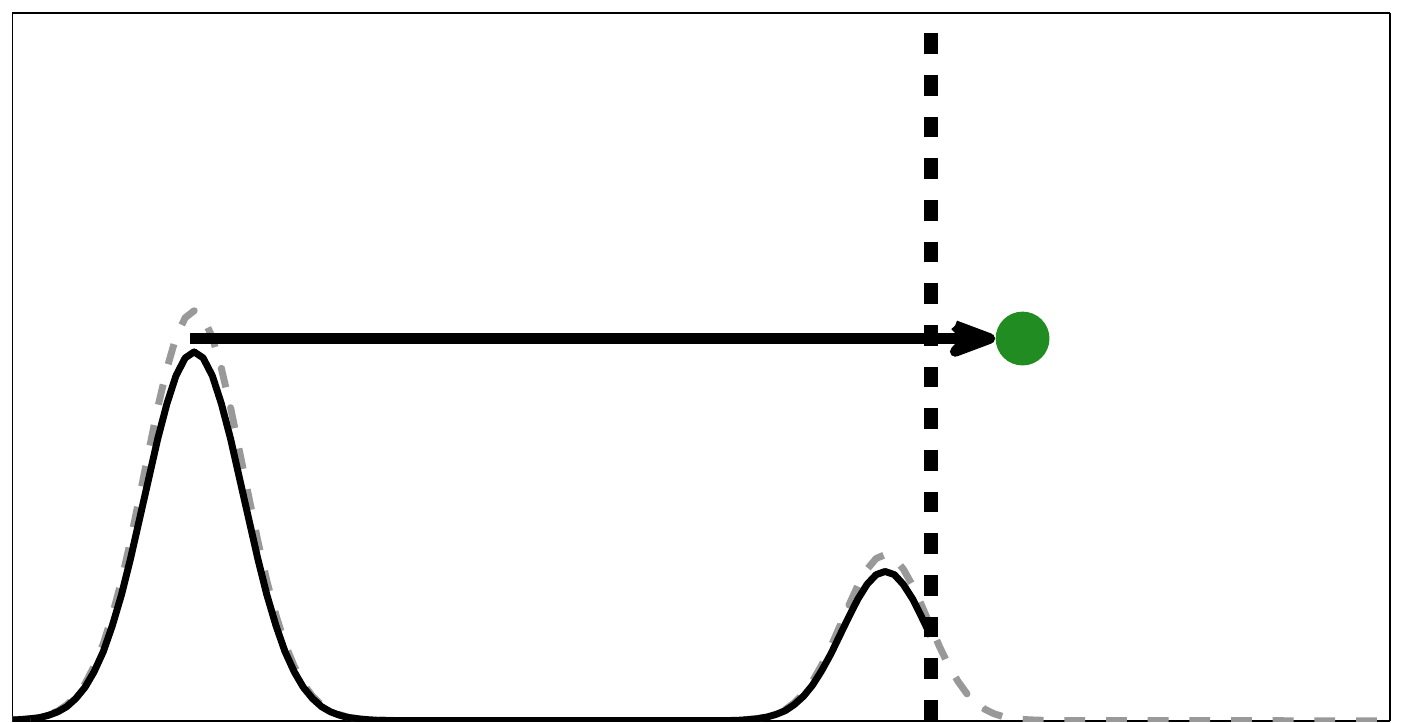}
                \put(30, 2){$\Omega_P$}
                \put(85, 2){$\Omega_B$}
                \put(67, 10){$I$}
                \put(77, 26){\color{forestgreen}$x_i(2\Delta t)$}
                \put(50, -3){$x$}
                \put(25, 45){\color{grey} $\widetilde{p}(x, 2\Delta t)$}
                \put(25, 37){$p(x, 2\Delta t)$}
                \put(11, 2){98}
                \put(61, 2){1}
            \end{overpic}
        }
        \caption{Thought experiment that leads to errors in the variance.}
        \label{fig:thoughtexp}
    \end{figure}
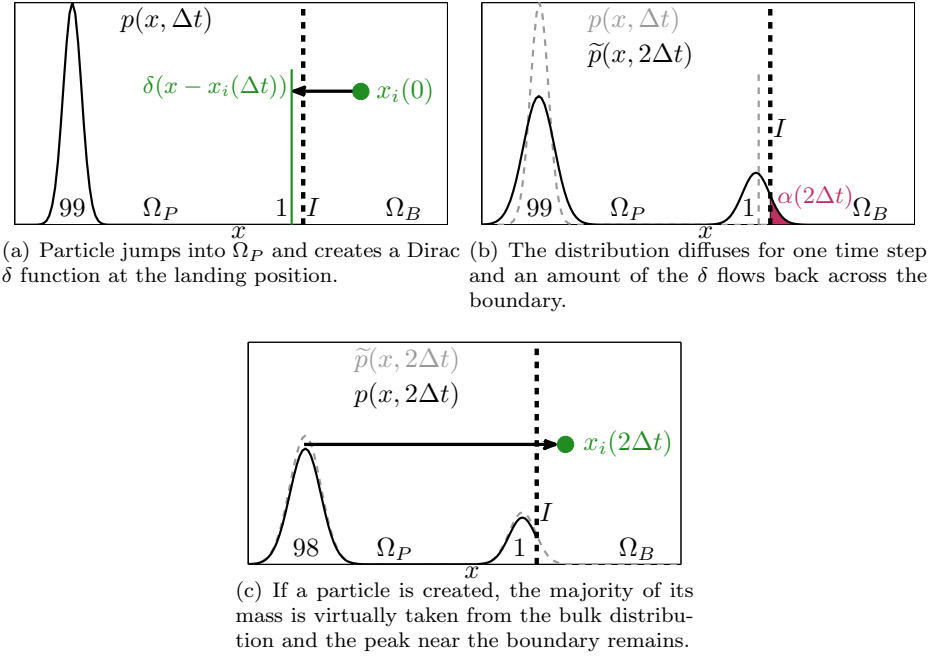

Of course, our thought example is an extreme case: we would expect
that in practical cases there would be a significant density of
particles throughout the continuum region (otherwise we would be
tracking them individually). Nevertheless the fact that all
information about an individual  
particle is lost as soon as it crosses the interface does
generate an error in the variance of particle numbers near the
interface, and the  effect 
becomes more pronounced when the concentrations in 
$\Omega_P$ close to $I$ are low.

\subsubsection{Dependence of the variance on the system parameters}
\label{subsubsec:varianceerrors}
We want to quantify the error in the variance as a function of 
the system parameters, which are the size of
the domain $[-L,L]$, the diffusion constant $D$, the simulated 
time $T_{\mathrm{final}}$ and the total mass $N$. 
After a nondimensionalisation the macroscopic PDE (\ref{eq:diffusion})
can be written in the form
\begin{equation*}
\frac{\partial n}{\partial t}
= 
\frac{\partial^2 n}{\partial x^2}\,,\qquad x\in[-1, 1]\,,
\end{equation*}
where the simulation is run until
\begin{equation*}
T_{\mathrm{final}}^* = T_{\mathrm{final}} \frac{L^2}{D}\,.
\end{equation*}
Hence, the system only has two parameters that need to be investigated: 
the final simulation time $T_{\mathrm{final}}^*$ and the total 
number of molecules $N$.

As in Figure \ref{fig:example1}, we use PBD algorithm (A.1)--(A.5)
with $\Omega_P = (-1,0)$, $\Omega_B = (0,1)$ and $I = \{0\}$.
For each parameter ($T_{\mathrm{final}}^*$ and $N$), we simulate 
the system 1000 times for different values of this parameter and 
measure in each case the number of particles in $\Omega_B$ at the end 
of the simulation, i.e. the value
$N_B(T_{\mathrm{final}}^*)$. In Figure~\ref{fig:varianceerror},
\begin{figure}
\center
\subfigure[$N = 10,20,\dots,100, \quad T_{\mathrm{final}}^{*} = 0.2$]{
\epsfig{file=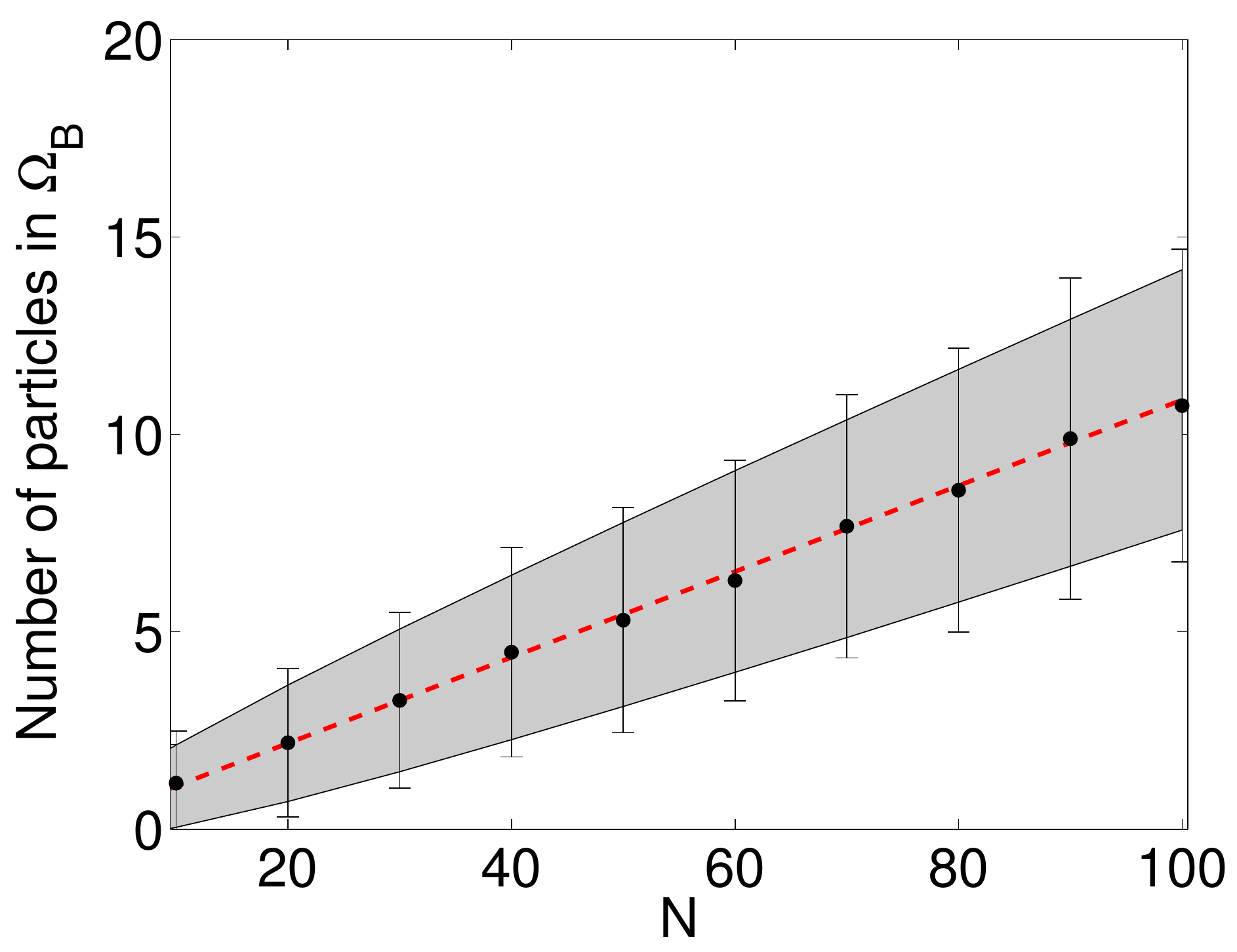, width=0.45\textwidth}
\label{fig:varianceerror:N2}
}
\subfigure[$T_{\mathrm{final}}^* = 0.01,0.02,\dots,0.2, \quad N=100$]{
\epsfig{file=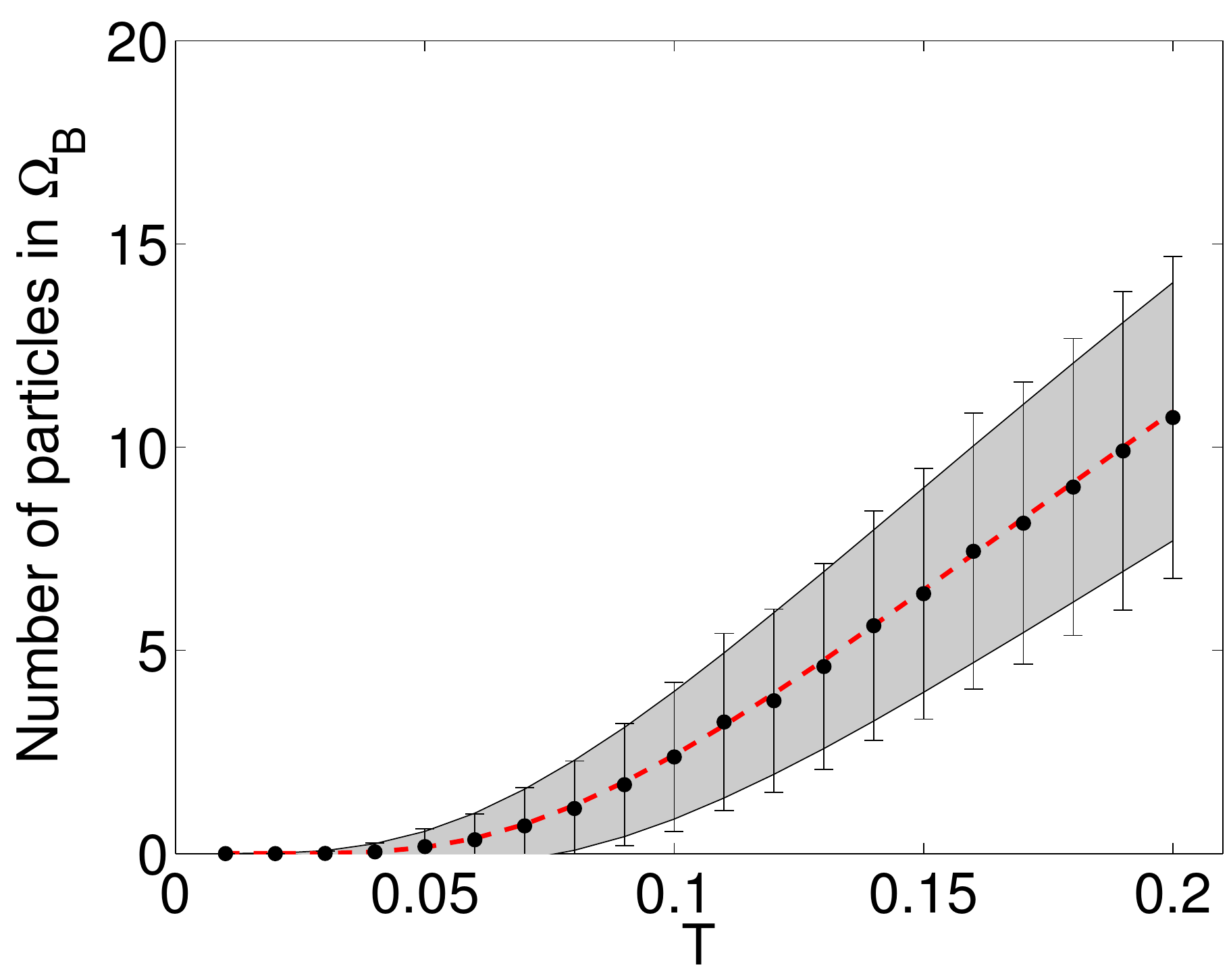, width=0.45\textwidth}
\label{fig:varianceerror:T}
}
\caption{Mean values and standard deviations of the number of 
$N_B(T_{\mathrm{final}}^*)$ depending on 
$T_{\mathrm{final}}^*$ and $N$. 
Dashed line: expected number of particles; shaded area: expected standard 
deviation; dots: measured mean values; error bars: 
measured standard deviations. Other parameters are chosen
as for Figure~\ref{fig:example1}.}
\label{fig:varianceerror}        
\end{figure}
we see that the measured mean values match well with the expected 
outcomes. For the standard deviations, however, we see that for all 
values of $T_{\mathrm{final}}^*$ and $N$ the 
measured outcomes are higher than expected. This is an undesired
effect and the next section will discuss a way to reduce this 
artefact.

\section{A PBD algorithm with an overlap region} 
\label{sec:probation}
In Section~\ref{subsec:problems} we saw that the immediate return 
of particles from $\Omega_P$ into $\Omega_B$ in combination with 
relatively low concentrations close to the interface can lead to 
errors in the variance of particle concentrations 
in $\Omega_B$. One way to overcome this problem is the 
introduction of an overlap region where BD 
simulation and a continuum description exist in parallel,
i.e. we will consider the case [B] defined in Section
\ref{sec:formulation} by $\Omega_B \cap \Omega_P \not = \emptyset$.
A sketch of this new set up can be seen in 
Figure~\ref{fig:sketch:probation} with the overlap region
denoted as $O = \Omega_B \cap \Omega_P$. We also denote the interfaces
$I_1$ and $I_2$ by
\begin{equation*}
I_1 = \partial \Omega_B \cap \Omega_P, \qquad
I_2 = \Omega_B \cap \partial \Omega_P.
\end{equation*}
In the case of  diffusion only,
this new setup requires only subtle changes in the algorithm. 
Molecules are now incorporated into the concentration 
when they cross the boundary $I_1$. The definition 
of $\widetilde{p}$ is still equal to (\ref{eq:def:ptilde}),
but we have to redefine $\alpha(t+\Delta t)$ and
$p_2(x, t+\Delta t)$ as follows
\begin{eqnarray} \label{eq:def:alpha2}
\alpha(t + \Delta t) 
&=& 
\int_{\Omega_B\setminus\Omega_P} 
\widetilde{p}(x, t + \Delta t) \, \dx\,,
\\
 \label{eq:p22}
p_2(x, t+\Delta t) &=& 
\frac{\widetilde{p}(x, t+\Delta t)}{\alpha(t+\Delta t)}\,,
\qquad \mbox{for} \; x \in \Omega_B \setminus \Omega_P \,.
\end{eqnarray}
The introduction of the overlap region prevents undesired 
effects generated by molecu\-les crossing over and coming back 
straight away, as discussed in the thought experiment
in Section \ref{subsec:problems}. 
In particular, a molecule initialized as a Dirac 
$\delta$ function in $\Omega_P \setminus \Omega_B$ initially
contributes very little to the overall 
probability density near the interface $I_2$; by the time it has a
significant probability of crossing $I_2$ its distribution 
has become sufficiently spread that it is  `lost' in the subdomain
$\Omega_{P}$ as is required for the continuous 
distribution. One time step of the second PBD algorithm is presented in 
Table~\ref{alg:probation}. As before, we consider that the time step 
$\Delta t$ is chosen so small that $\alpha(t + \Delta t) \ll 1$.
Therefore, we only need to implement cases (i)--(ii) in step (B2), 
because the probability that two or more molecules are initiated 
in $\Omega_B$ during one time step is negligible.
\begin{figure}
 \centering
        \begin{overpic}[width=0.6\textwidth]{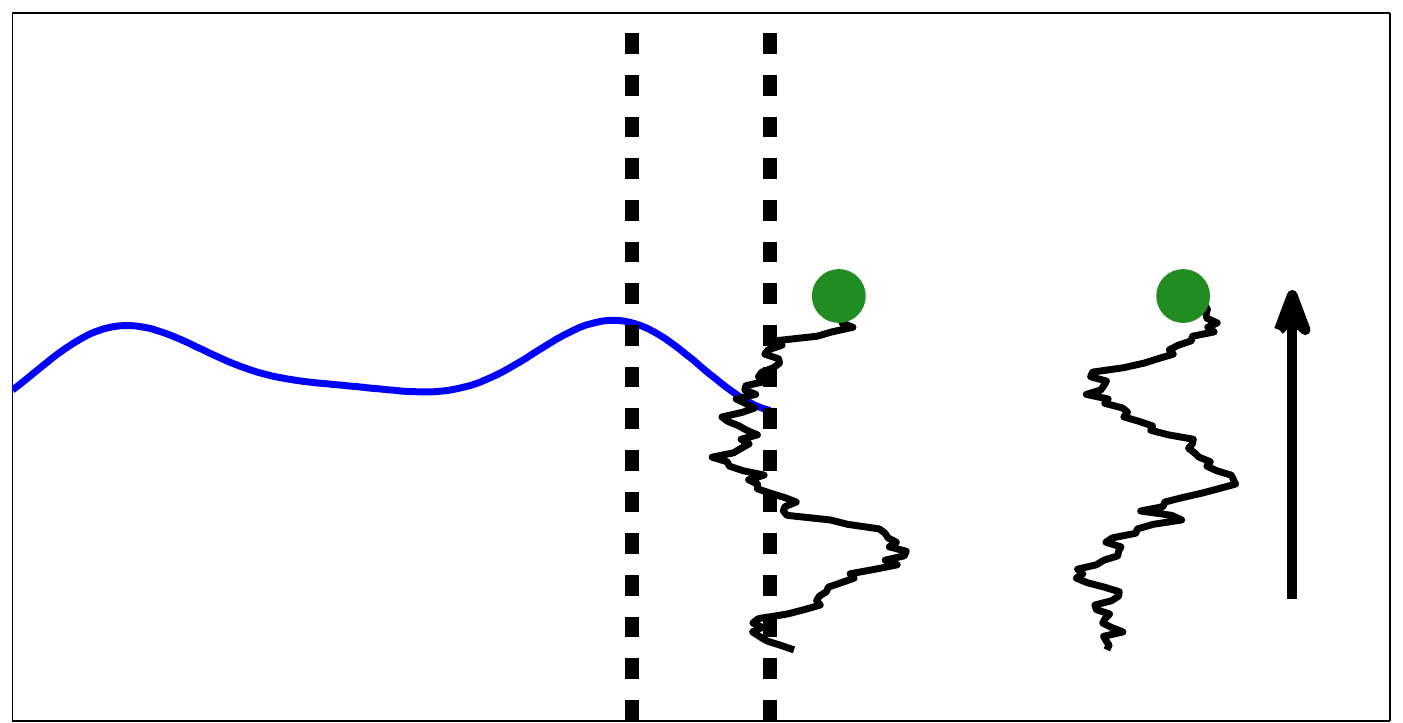}
            \put(20, 2){\Large $\Omega_P$}
            \put(70, 2){\Large $\Omega_B$}
            \put(47, 2){\Large $O$}
            \put(39, 2){\Large $I_1$}
            \put(56, 2){\Large $I_2$}
            \put(20, 45){\color{blue}\Large $p(x, t)$}
            \put(70, 45){\Large $N_B(t)$}
            \put(57, 34){\color{forestgreen}\large $x_i(t)$}
            \put(80, 34){\color{forestgreen}\large $x_j(t)$}
            \put(94, 20){\large $t$}
            \put(50, -5){\Large $x$}
        \end{overpic}
        \caption{Sketch of the PBD algorithm with overlap region 
	and the notation related to it. 
        In $\Omega_P$, molecules are described by their density 
	distribution $p(x, t)$,
        in the microscopic domain $\Omega_B$ described by the 
	number $N_B(t)$ of molecules and their positions $x_i(t)$,
        $i=1,\dots,N_B(t)$. In the 
	overlap region $O$, both descriptions
        exist in parallel. The interfaces between the various 
	subdomains are denoted $I_1$ and $I_2$.}
        \label{fig:sketch:probation}
    \end{figure}
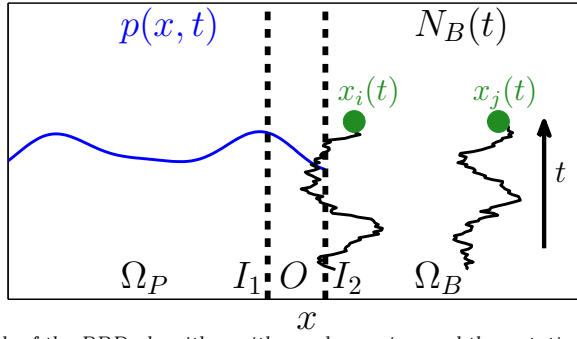

\renewcommand{\labelenumi}{(B\arabic{enumi})}
\renewcommand{\labelenumii}{\textbf{(\roman{enumii})}}
\begin{table}[t]
\begin{framed}
\begin{enumerate}
    \setcounter{enumi}{0}
        \item Calculate $\widetilde{p}(x, t+\Delta t)$ using 
		\eqref{eq:def:ptilde} and $\alpha(t+\Delta t)$
                using \eqref{eq:def:alpha2}.
        \item Generate uniformly distributed random number $r$ in $(0,1)$. 
        \begin{enumerate}
           \item If $r < \alpha(t+\Delta t)$, then create new particle 
	   in $\Omega_B \setminus \Omega_P$ according to the probability 
	   density
           $p_2(x, t+\Delta t)$ defined in \eqref{eq:p22}. Set
	   $\beta = \beta_{(i)}$ where $\beta_{(i)}$ is given
	   by (\ref{eq:defn:betas}). Set $N_{B,1} = 1.$
           \item If $r \geq \alpha(t+\Delta t)$, then set
	   $\beta = \beta_{(ii)}$ where $\beta_{(ii)}$ is given
	   by (\ref{eq:defn:betas}). \\ Set $N_{B,1} = 0.$
	\end{enumerate}

        \item Compute positions $x_i(t+\Delta t)$, $i=1,2,\dots,N_B(t)$, 
	      of BD particles according to \eqref{eq:BD}. 
        \item Compute new concentration in $\Omega_P$ by
            \begin{equation*}
                p(x, t+\Delta t) = \beta \,
		 \widetilde{p}(x, t+\Delta t) 
		 + \sum_{x_i(t+\Delta t) \in \Omega_P \setminus \Omega_B} 
		 \!\!\!\!\!
		 \delta(x - x_i(t + \Delta t))\,, \quad 
		 \mbox{for} \; x\in\Omega_P\,.
            \end{equation*}
         \item  Update the number of BD particles by  
	  \begin{equation*}
                N_B(t+\Delta t) = N_B(t) 
		+ 
		N_{B,1} 
		- 
		\left|\left\{ x_i(t+\Delta t)\in\Omega_P\setminus \Omega_B\;, 
		i=1,\dots,N_B(t)
		\right\}\right|
                \,.
            \end{equation*}
	    Terminate computation of trajectories of BD molecules which 
	    landed in $\Omega_P\setminus \Omega_B$ (i.e. the BD particles 
	    which satisfy $x_i(t+\Delta t)\in\Omega_P\setminus \Omega_B.$) \\
	Then continue with step (B1) for time $t+\Delta t$.
\end{enumerate}
\end{framed}
\caption{One time step of the PBD algorithm with overlap region
for system of diffusing molecules.}
\label{alg:probation}
\end{table}

In order to highlight the advantages of the overlap region, 
we simulate the same diffusion process as in 
Figure~\ref{fig:example1} with 
$\Omega_P = (-1,0)$ and $\Omega_B = (-0.1,1)$. Then
the overlap region is $O = \Omega_P \cap \Omega_B = (-0.1, 0)$. 
The results are shown in Figure~\ref{fig:example1_pz}.
As before, the mean outcome matches well with the exact solution, 
with stochastic fluctuations inside the overlap region $O$ due to the
mixed description. In Figure~\ref{fig:example1_pz:var} we
see that the introduction of the overlap region indeed reduced
the problem of high variances inside $\Omega_B \setminus \Omega_P$.
\begin{figure}
\center
\subfigure[Mean]{
\epsfig{file=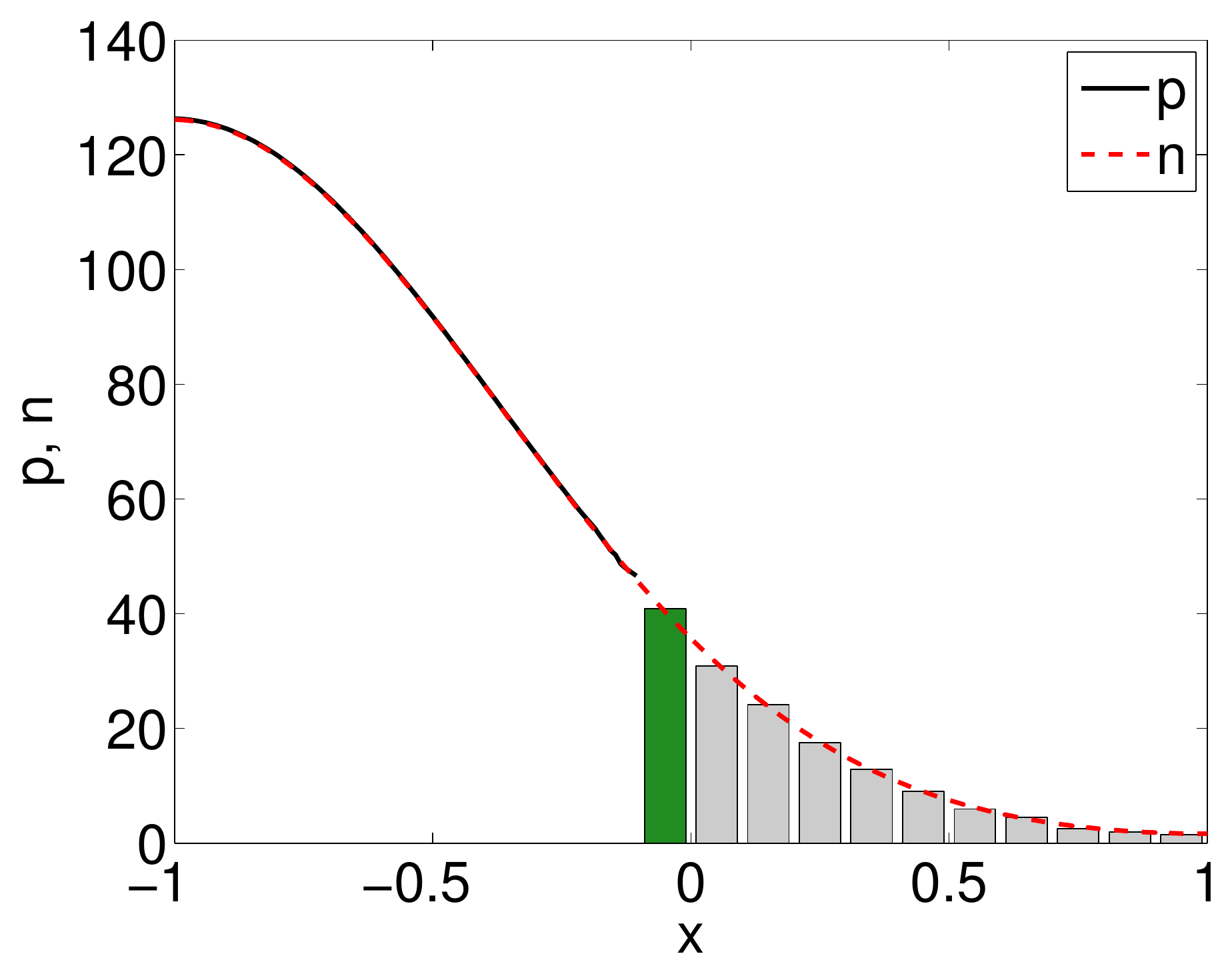, width=0.45\textwidth}
}
\subfigure[Variance]{
\epsfig{file=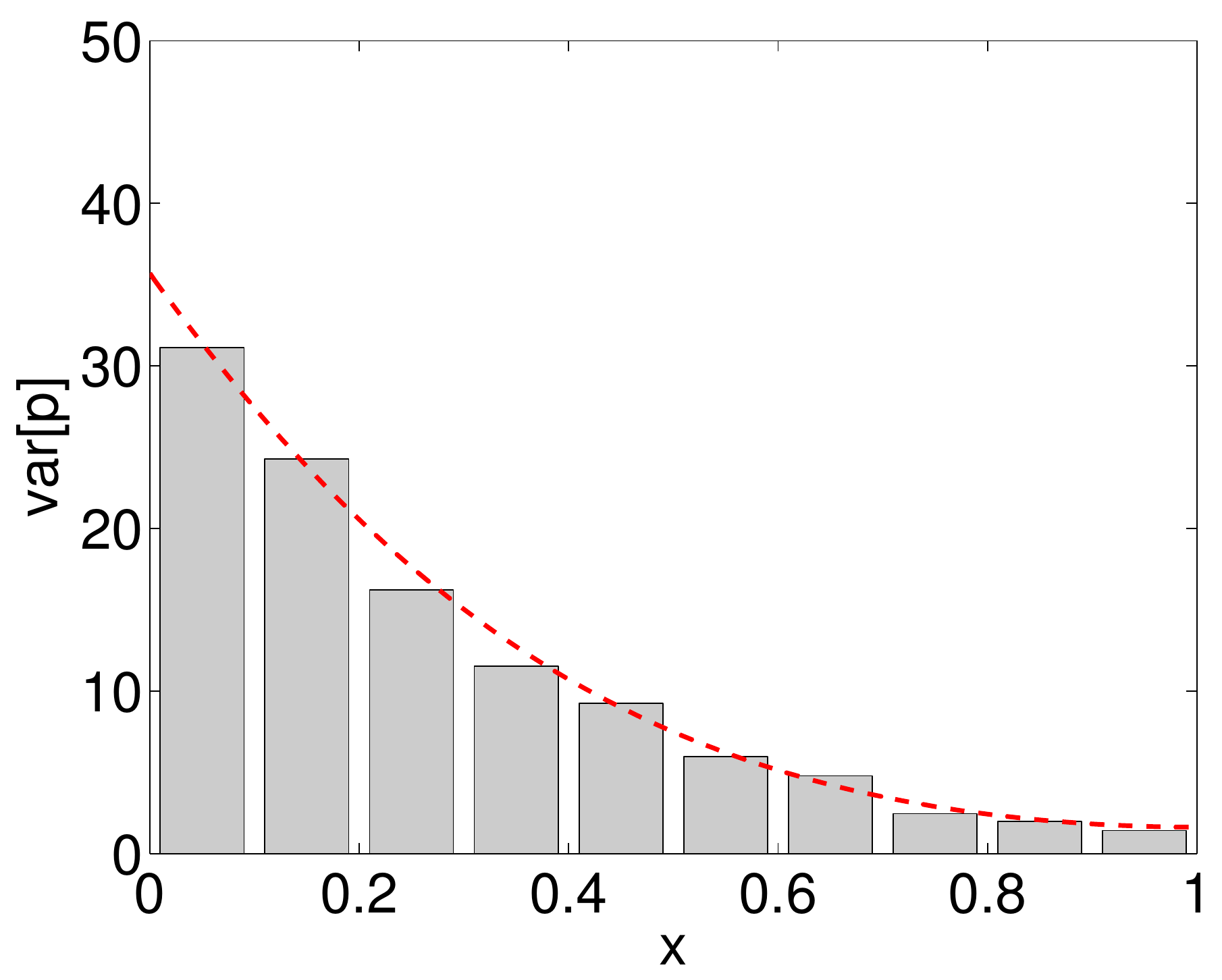, width=0.45\textwidth}
\label{fig:example1_pz:var}
}
\caption{Simulation results of a diffusion process in 
$\Omega = (-1,1)$ with no flux boundary conditions and 
initial conditions $n(x, 0) = 100\delta (x+0.95)$ with 
$\Omega_P = (-1, 0)$, $\Omega_B = (-0.1, 1)$ and
$O = (-0.1, 0)$ averaged over 1000 realisations. 
Dashed (red) line: expected outcome. {\rm (a)} Solid line: 
mean value in $\Omega_P$; bars: particle concentrations 
in $\Omega_B \setminus  O$ (gray) and $O$ (green). 
{\rm (b)} (Gray) bars: measured variances in particle concentrations. 
Parameter values as for Figure~\ref{fig:example1}.}
\label{fig:example1_pz}
\end{figure}
A proof similar to Theorem~\ref{thm:1} can be used to show
that the PBD algorithm (B1)--(B5) satisfies Conditions
(C.1) and (C.3) for $\Delta t > 0$. We will now further 
show that the algorithm describes a diffusion 
process exactly in the limit $\Delta t \to 0$.
    
\begin{theorem} \label{thm:2}
Suppose that a PDE-description of the system is 
used in $\Omega_{P} = (-\infty, 0)$, 
and a BD simulation in $\Omega_{B} = (-d, \infty)$ 
with the overlap region $O = (-d,0)$ where $d>0$.
In the limit that $\Delta t \rightarrow 0$ the expected
concentration distribution $P_P(x,t) = {\mathbb E}[p(x,t)]$  
in the continuum regime and the expected concentration 
distribution $P_B(x,t)$ in the BD regime obey the equations
\begin{align}
\label{gov1} 
\frac{\partial P_P}{\partial t} 
&= D\frac{\partial^2 P_P}{\partial x^2} + 
\left.
D\frac{\partial P_B}{\partial x}\right|_{x \rightarrow -d_+}
\delta\left(x+d\right),  
\quad \ x\in \Omega_{P}\,,\\
\label{gov2} \frac{\partial P_B}{\partial t} 
&= D\frac{\partial^2 P_B}{\partial x^2} - 
\left.
D\frac{\partial P_P}{\partial x}\right|_{x \rightarrow 0_-}
\delta\left(x\right),   
\quad \ x\in \Omega_{B}\,,
\end{align}
where 
$$
P_P(0,t) = 0\,, \qquad P_B(-d,t) = 0\,, \qquad \mbox{for} \quad t > 0.
$$
Extend each distribution to the whole line by defining $P_P(x,t) = 0$
for $x \in (0,\infty)$ and $P_B(x,t) = 0$ for $x \in (-\infty,-d)$. 
Then the sum of these two processes $n(x,t) = P_P(x,t) + P_B(x,t)$, 
$x\in\Omega$, $t>0$ satisfies the diffusion equation
\eqref{eq:diffusion}.
\end{theorem}
    
\begin{proof}
Consider the function 
\begin{equation} 
n(x,t) = P_P(x,t) + P_B(x,t).
\end{equation}
Clearly
\[  \frac{\partial n}{\partial t} 
	    = D\frac{\partial^2 n}{\partial x^2} \qquad \mbox{ for }x
            \in 
	        (-\infty,-d)\cup(-d,0)\cup(0,\infty).
\] 
Since each of $P_P$ and $P_B$ is continuous at $x=-d$ and $x=0$, function
$n$ will be continuous there. Moreover, since
\[ \left[ \frac{\partial P_P}{\partial x} \right]^{x=0_+}_{x=0_-} =
-\frac{\partial P_P}{\partial x}\left(0_-,t\right), \qquad
 \left[ \frac{\partial P_B}{\partial x} \right]^{x=0_+}_{x=0_-} =
\frac{\partial P_P}{\partial x}\left(0_-,t\right),\]
$\partial n/\partial x$ is continuous at $x=0$. Similarly, since
\[ \left[ \frac{\partial P_P}{\partial x} \right]^{x=-d_+}_{x=-d_-} =
-\frac{\partial P_B}{\partial x}\left(-d_+,t\right), \qquad
 \left[ \frac{\partial P_B}{\partial x} \right]^{x=-d_+}_{x=-d_-} =
\frac{\partial P_B}{\partial x}\left(-d_+,t\right),\]
$\partial n/\partial x$ is also continuous at $x= -d$.
By standard regularity results this is enough to guarantee that $n$
satisfies the diffusion equation on the whole real line.
\end{proof}

\noindent\textbf{Remark.} For the overlap region to 
give a different result from the simple PDB algorithm
(A1)--(A5), we should choose $d$ much bigger than the mean
displacement $\sqrt{2D\Delta t}$ of a particle  
given the time step $\Delta t$ as defined in \eqref{eq:BD}.

\subsection{Dependence of the variance on the system parameters}
In order to show that the variances are indeed accurately
produced by the PBD algorithm (B1)--(B5), we repeated the numerical 
experiments conducted in Section~\ref{subsubsec:varianceerrors}. 
We chose $\Omega_P = (-1, 0)$, $O = (-0.1, 0)$ and $\Omega_B = (-0.1, 1)$ 
and measure the number molecules situated in $\Omega_B\setminus\Omega_P
= [0,1)$ 
at the end of the simulation. We again calculate the
mean values and standard deviation and present the results in 
Figure~\ref{fig:varianceerror:PZ}.
\begin{figure}
\center
        \subfigure[$N = 10,20,\dots,100, \quad T_{\mathrm{final}}^* = 0.2$]{
            \epsfig{file=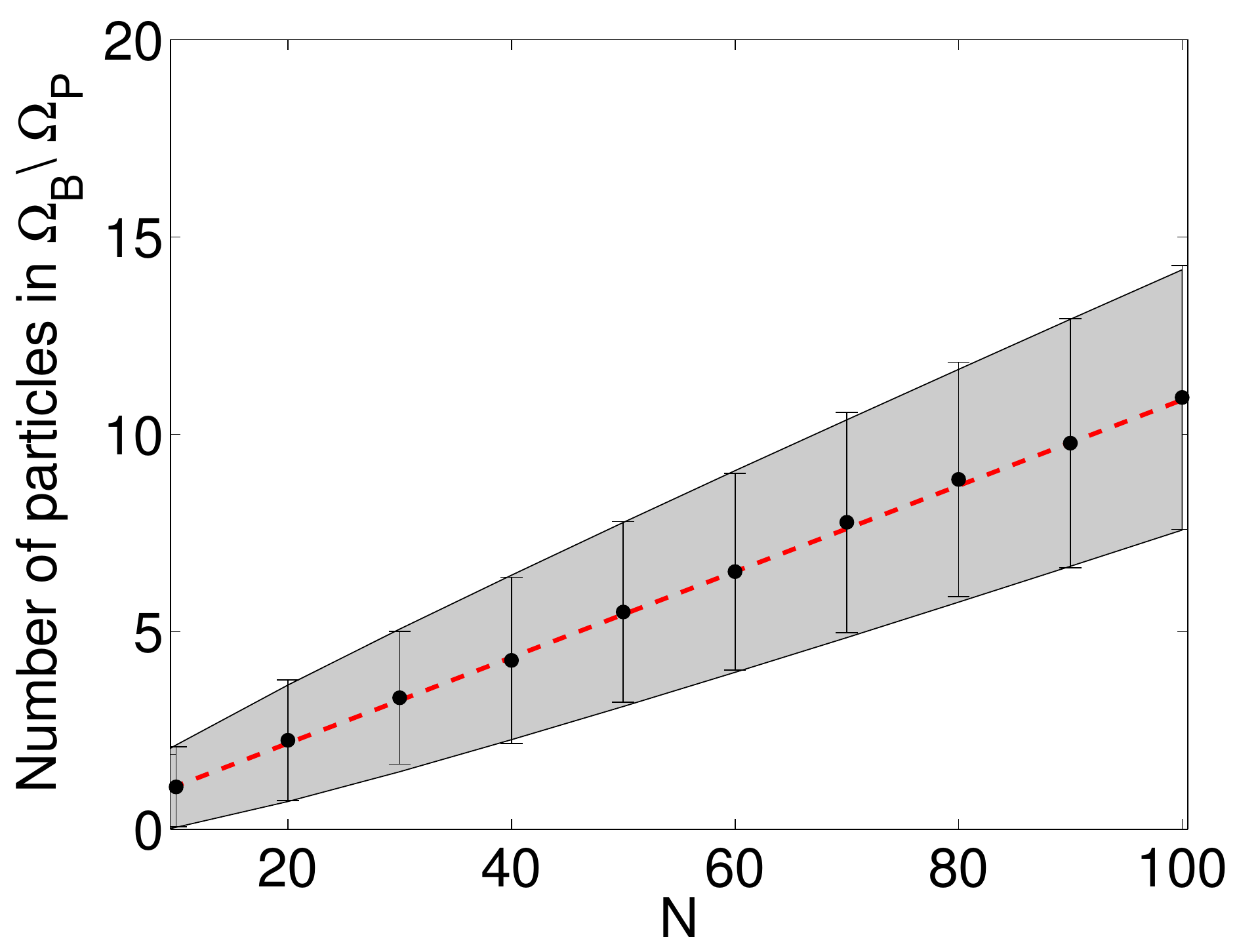, width=0.45\textwidth}
            \label{fig:varianceerror:N2:PZ}
        }
        \subfigure[$T_{\mathrm{final}}^* = 0.01,0.02,\dots,0.2, \quad N=100$]{
            \epsfig{file=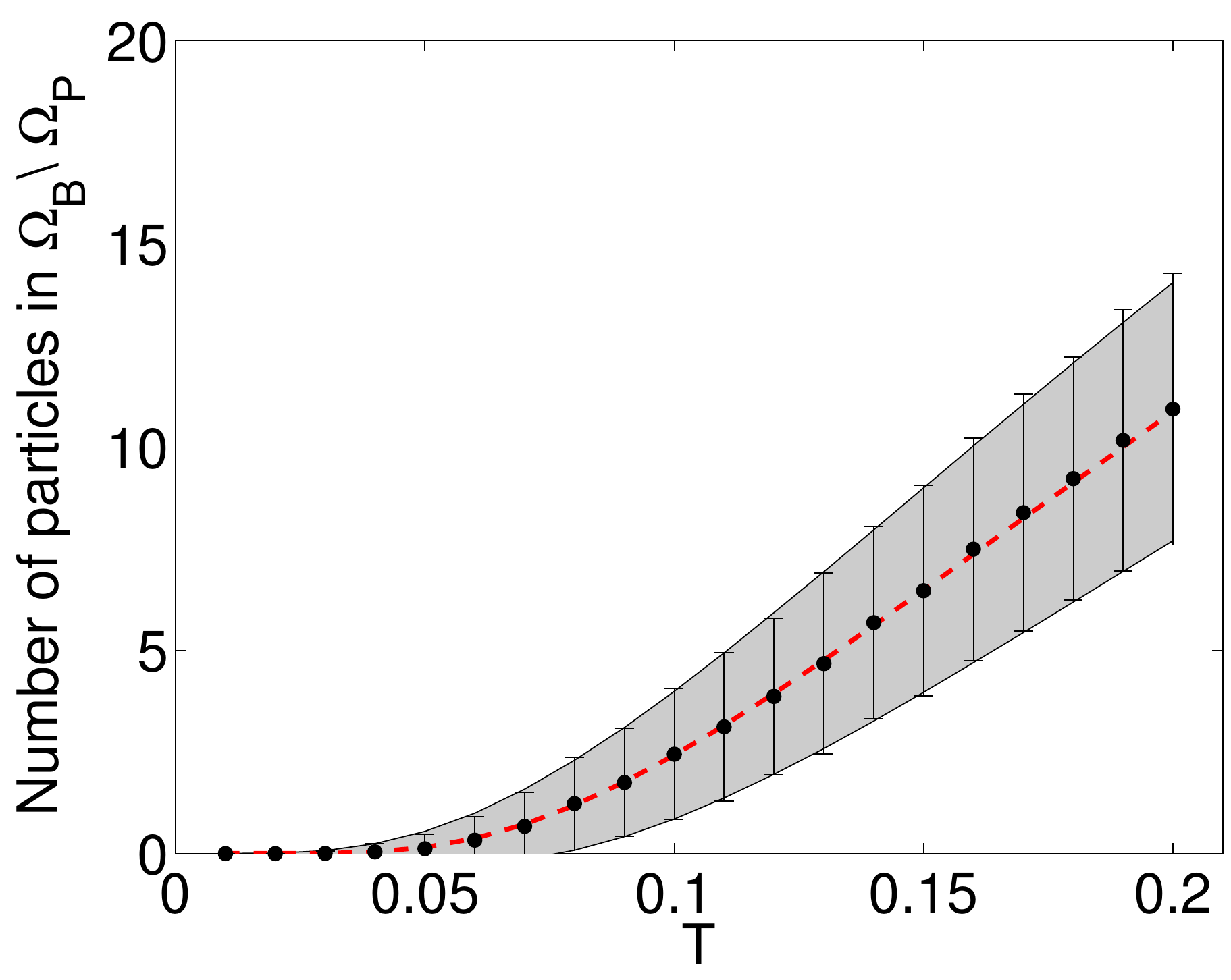, width=0.45\textwidth}
            \label{fig:varianceerror:T:PZ}
        }
        \caption{Mean values and standard deviations of the number of 
	particles in $\Omega_B\setminus\Omega_P$ at time 
	$T_{\mathrm{final}}^*$ 
        depending on $T_{\mathrm{final}}^*$ and $N$. 
	Dashed line: expected number of particles; shaded area: 
	expected standard 
        deviation; dots: measured mean values; error bars: 
	measured standard deviations. Parameters as for 
	Figure~\ref{fig:example1_pz}.}
        \label{fig:varianceerror:PZ}
\end{figure}
We can clearly see that the adjusted algorithm produces more accurate 
standard deviations than the original
algorithm, whilst keeping the mean values correct. We 
conclude that the algorithm (B1)--(B5) produces an accurate
BD simulation inside $\Omega_B\setminus\Omega_P$ and therefore satisfies the
Conditions (C.1)--(C.3).
We will use the PBD algorithm (B1)--(B5) in the remainder 
of this paper.

\section{Reaction-diffusion systems} \label{sec:reactions}
In the next step we introduce chemical reactions into the system 
presented in Section~\ref{sec:probation}. We will concentrate 
on zero-order and first-order reactions \cite{Erban:2007:PGS}. 
First-order reactions are reactions which only have one 
reactant, for example, 
\begin{equation}
X_1 \ \stackrel{k_1}{\longrightarrow} \ X_2 \,,
\qquad
\mbox{or}
\qquad
X_1 \ \stackrel{k_2}{\longrightarrow} \ X_4 + X_5 \,,
\label{pomreact}
\end{equation}
where $X_i$ denote chemical species and $k_1$ (resp. $k_2$) 
the corresponding rate constant (which has physical units
$[$sec$^{-1}]$). In what follows, we will denote by $\emptyset$ 
 chemical species which are
of no interest to a modeller. Then, considering that
$X_1$ is the only chemical species of interest, we can
rewrite the reactions (\ref{pomreact}) as
\begin{equation} 
\label{eq:reactionsdeg}
X_1 \ \stackrel{k_d}{\longrightarrow} \ \emptyset,
\end{equation}
where $k_d = k_1 + k_2$. We will also consider zero-order
reactions. An example is:
\begin{equation} 
\label{eq:reactionspro}
\emptyset \ \stackrel{k_p}{\longrightarrow} \ X_1,
\end{equation}
where the rate constant $k_p$ has physical units $[$M$\,$sec$^{-1}]$,
i.e. it is the production rate per unit of volume and unit of time.
It is relatively straightforward to implement zero-order and 
first-order chemical reactions in the PBD algorithms (A1)--(A5)
and (B1)--(B5), because these reactions can be treated in the
individual parts of the system (continuum and BD simulation) 
independently. Note that for higher-order reactions this 
is not necessarily the case, as particles could react with the 
continuum inside the overlap region $O$. This case is not discussed
in this paper.
    
In the continuum regime reactions are represented by the term 
$R_j(p_1, p_2, \dots, p_M)$ on the right-hand side of the 
reaction-diffusion PDE \eqref{eq:reactdiff}. For example, if the chemical
species $X_1$ is subject to chemical reactions 
(\ref{eq:reactionsdeg})--(\ref{eq:reactionspro}), then the
reaction-diffusion PDE \eqref{eq:reactdiff} takes the form
\begin{equation*}
\frac{\partial p_1}{\partial t} 
= 
D_1 \frac{\partial^2 p_1}{\partial x^2} 
- k_d \, p_1 + k_p \,.
\end{equation*}
In the BD simulations, the molecules act independently and the 
reactions can therefore be treated individually. A summary 
of how to implement various reactions in BD simulations can 
be found in \cite{Erban:2007:PGS}. 

Although implementation of the zero-order and first-order
reactions is relatively straighforward, one has to 
still consider some special effects that are related to the 
coupling of the two parts of the domain. 
First, what happens when a particle that is supposed to 
react at time $t_1$ crosses the interface $I_1$ at an earlier time
$t_2 < t_1$? Since we assumed that all information about 
particles is lost as soon as they cross the interface $I_1$, 
we incorporate it into the continuum and the reaction 
at time $t_1$ does not happen.
Second, what happens when a particle is created inside 
the overlap region $O$? A number of solutions to this problem 
are possible: one could split the creation in equal parts, 
or declare creation to only contribute to either the continuum 
or the molecular-based description. We will here assume that 
all creation inside the overlap region occurs in the form 
of molecules with exact positions. 

Finally, let us note that (reactive) boundary conditions on the 
external boundary $\partial \Omega$ can be treated according to 
the corresponding modelling regime, i.e. whether the corresponding
segment of $\partial \Omega$ is part of $\partial \Omega_P$ or 
$\partial \Omega_B$. Derivation of reactive boundary conditions 
of BD simulations which are consistent with the PDE description 
can be found in \cite{Erban:2007:RBC}. External boundaries slightly 
modify the computation of $\widetilde{p}(x, t+\Delta t)$ in 
$\Omega_P$. It is still given by (\ref{eq:def:ptilde}) but the kernel
(\ref{eq:def:kernel}) has to be updated to take into account
the boundary condition imposed on the external boundary
$\partial \Omega$. We have already
done this when we showed simulations of the diffusion process
in Figures~\ref{fig:example1} and \ref{fig:example1_pz} in the 
finite interval $(-1,1)$ with no flux boundary conditions. However,
for small timesteps $\Delta t$ the change is negligible, since all the
action takes place near the interface $I_2$.

We conclude this section with three examples which illustrate the 
behaviour of the PBD algorithm (B1)--(B5) for reaction-diffusion
systems. They include the modelling of morphogen gradients and 
chemisorption.

\subsection{Example 1: morphogen gradient}
In the first example we compute a steady state for a morphogen 
gradient model \cite{Tostevin:2007:FLP,Wolpert:2002:PD,Howard:2012:HBR}.
We consider one chemical species (morphogen) inside the
domain $\Omega = (-1, 1)$. All parameters are dimensionless
for simplicity. The only reaction inside $\Omega = (-1, 1)$ is the 
degradation (\ref{eq:reactionsdeg}). Additionally to this reaction, 
we assume a constant influx 
$J/D$ through the left-hand boundary $x=-1$ to the
continuum subdomain $\Omega_P= (-1, 0)$. 
We use $O = (-0.1, 0)$ and $\Omega_B = (-0.1, 1)$ with a no 
flux boundary at $x=1$. Since we only have first-order reaction
(\ref{eq:reactionsdeg}), the exact solution $n(x,t)$ is given by
\begin{equation} \label{eq:exp1:full}
  \frac{\partial n}{\partial t} = D \frac{\partial^2 n}{\partial x^2} - k_1 n\,,\qquad
  \frac{\partial n}{\partial x} (-1, t) = -J\,, \qquad
  \frac{\partial n}{\partial x} (1, t) = 0\,.
\end{equation}
This system is initialised with $n(x, 0) = 0$, $x \in \Omega$, 
and we let it run until it (approximately) reaches the steady state.  
We use $J=1000$ and $k_1 = 1$. 
The second PBD algorithm (B1)--(B5) is run with the same 
parameters presented in Section~\ref{sec:probation} for time 
$T_{\mathrm{final}} = 20$ which is (approximately) a time at 
which the model settles into the steady state. 
The reaction \eqref{eq:reactionsdeg} is simulated in a time-driven 
manner in $\Omega_B$, which means that for each 
morphogen molecule it is decided randomly at the end 
of each time step whether it was degraded or not
(the probability of degradation of each molecule is equal 
to $k_d \, \Delta t$ provided that $k_d \, \Delta t \ll 1$).

The result of a single simulation of the PBD algorithm (B1)--(B5)
is plotted in Figure~\ref{fig:exp1:single}. We plot the
PDE solution in $\Omega_P \setminus O$ as a black line and
the (gray) histogram of molecules in $\Omega_B \setminus O$. 
In the overlap region $O$, we compute the total ``number" of
particles by
\begin{equation}
N_O(T_{\mathrm{final}})
\equiv
\int_O p(x,T_{\mathrm{final}}) \dx 
+ 
\big| 
\left\{ x_i(T_{\mathrm{final}})\in O\;, 
i=1,2, \ldots, N_B(T_{\mathrm{final}})\right\}
\big|.
\label{numberoverlap}
\end{equation}
The value of $N_O(T_{\mathrm{final}})$ is plotted as the green bar
in Figure~\ref{fig:exp1:single}. The results of a single simulation 
of the PBD algorithm (B1)--(B5) are compared with the exact solution,
which is obtained by solving the PDE (\ref{eq:exp1:full}) numerically
until time $T_{\mathrm{final}}$. In this simple example, it is also 
possible to find an analytical expression for the steady state
profile which is approximately equal to the presented dashed line.

Note that the jagged appearance of the continuum solution close to the
interface is not numerical error, but represents the fact that as
molecules cross from the discrete to the continuum side information
about their exact location is lost gradually over time (remember that
Figure~\ref{fig:exp1:single} shows just one realisation of the
stochastic process). The 
corresponding distribution on the discrete side $\Omega_B \setminus \Omega_P$
would be $\delta$ function spikes at the location of the particles, which we
have in effect locally averaged by our binning process.
Thus the jaggedness can be seen as a gradual transition in the
solution from isolated discrete particles to a continuum density distribution. 
An ensemble average over 100 simulations of the PBD algorithm 
(B1)--(B5) is shown in Figure~\ref{fig:exp1:avg100}. The stochastic 
fluctuations are reduced compared to the single simulation and the 
results compare well with the exact solution for the expected
probability density (\ref{eq:exp1:full}). 
\begin{figure}
\center
\subfigure[Single simulation]{
\epsfig{file=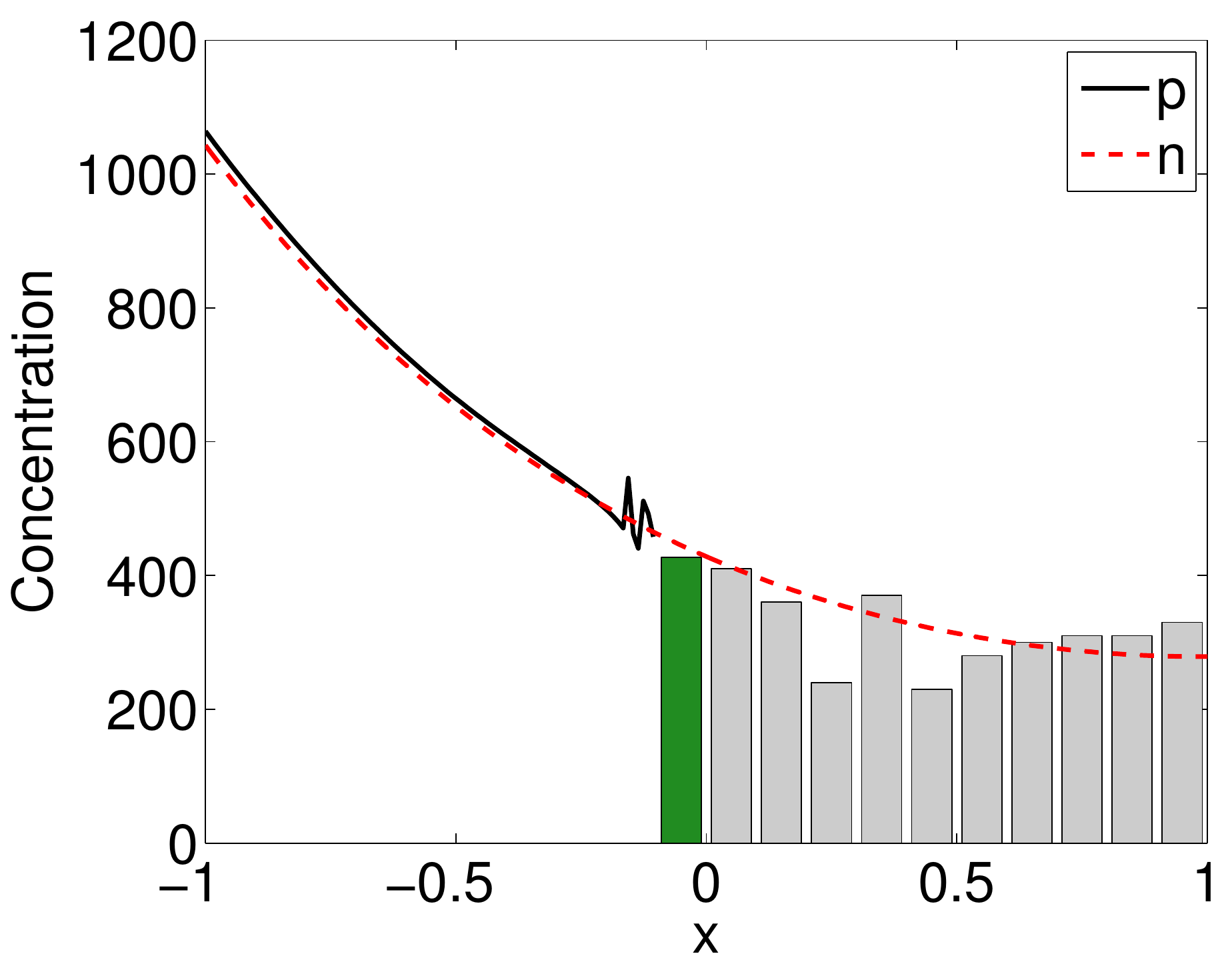, width=0.45\textwidth}
\label{fig:exp1:single}
}
\subfigure[Average over 100 realisations]{
\epsfig{file=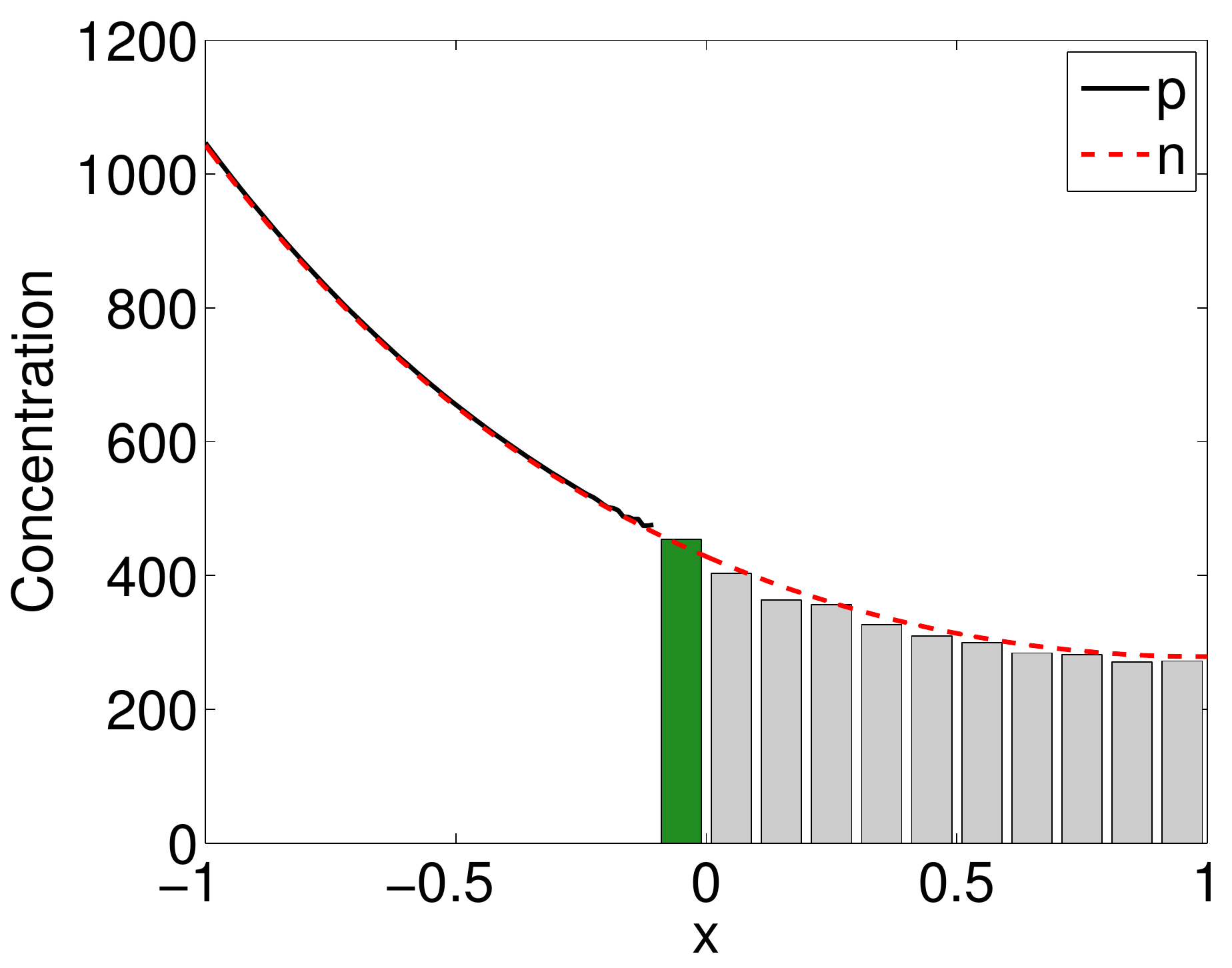, width=0.45\textwidth}
\label{fig:exp1:avg100}
}
\caption{Simulation results for Example 1. 
Dashed (red) line: exact solution given by $(\ref{eq:exp1:full})$; 
solid line: $p(x, T_{\mathrm{final}})$; 
(gray) bars: spatial concentration of particles
at $t=T_{\mathrm{final}}$ in $\Omega_B \setminus  O$; 
(green) bar: $N_O(T_{\mathrm{final}})$ given by $(\ref{numberoverlap})$.
Parameters as described in the text.}
\label{fig:exp1}
\end{figure}        

\subsection{Example 2: reversed morphogen gradient}
In this example we introduce a second reaction in addition to 
\eqref{eq:reactionsdeg} -- a local production of molecules:
\begin{equation} \label{eq:creation}
\emptyset \ \stackrel{k_p}{\longrightarrow} \ A \qquad \mbox{in } [x_s, 1]\,,
\end{equation}
where $x_s$ defines the size of the creation zone. As before
we consider all parameters to be  dimensionless: $k_p$ is defined as the
rate of production per unit length. For this system we use no 
flux boundary conditions on both ends. The combination of
localized production (\ref{eq:creation}) and degradation
(\ref{eq:reactionsdeg}) ensures that the system settles into 
a non-trivial steady-state which we will compute with the
PBD algorithm (B1)--(B5). The exact solution (which the
PBD algorithm (B1)--(B5) approximates) can be described
by 
\begin{equation} \label{eq:exp2:full}
\frac{\partial n}{\partial t} = D \frac{\partial^2 n}{\partial x^2} - k_1 n + k_2\chi_{[x_s, 1]} \,,\qquad
\frac{\partial n}{\partial x}(-1, t) = 0\,, \qquad
\frac{\partial n}{\partial x}(1, t) = 0\,,
\end{equation}
where $\chi_{[x_s, 1]}$ is the characteristic function for the
interval $[x_s, 1]$ that takes the value $1$ inside and $0$
outside of the interval. The production reaction \eqref{eq:creation}
was implemented in BD simulations in an event-driven way, such 
that particles can
get created at any time in-between time steps and the number of
particles created in one time step is not limited. We used $k_p =
1$, $k_d = 2000$ and $x_s = 0.5$. For the PBD simulations 
we use the same parameters as in Section~\ref{sec:probation}.
In particular, we have $\Omega_P= (-1, 0)$, $O = (-0.1, 0)$ 
and $\Omega_B = (-0.1, 1)$.

A single realisation of this process is plotted in Figure~\ref{fig:exp2:single}.
We plot the PDE solution in $\Omega_P \setminus O$ as a black line.
The concentration of molecules in $\Omega_B \setminus O$ is visualized
as a (gray) histogram. In the overlap region $O$, we plot 
$N_O(T_{\mathrm{final}})$ given by (\ref{numberoverlap}). 
The concentration gradient is now reversed, as the creation of particles 
happens near the right-hand boundary. Again, one can clearly see 
the stochastic fluctuations in this plot which also have effects 
on the value of $p(x,t)$ far from the overlap region $O$. However, 
as we draw an ensemble average over 100 simulations in 
Figure~\ref{fig:exp2:avg100}, the results converge
towards the exact solution which is obtained
by solving the PDE (\ref{eq:exp2:full}). It is plotted 
as a (red) dashed line in Figure~\ref{fig:exp2}.
\begin{figure}
\center
\subfigure[Single simulation]{
\epsfig{file=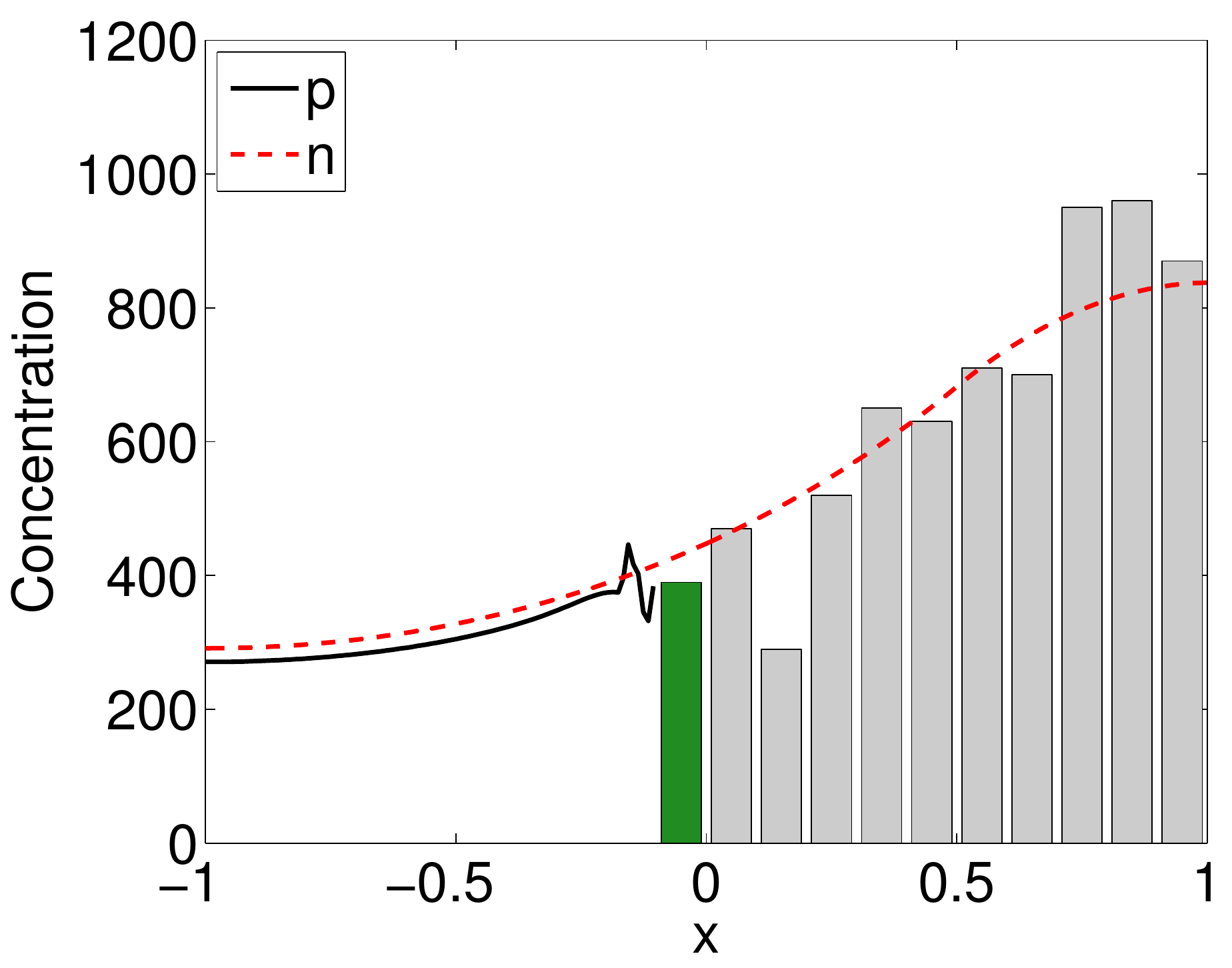, width=0.45\textwidth}
\label{fig:exp2:single}
}
\subfigure[Average over 100 simulations]{
\epsfig{file=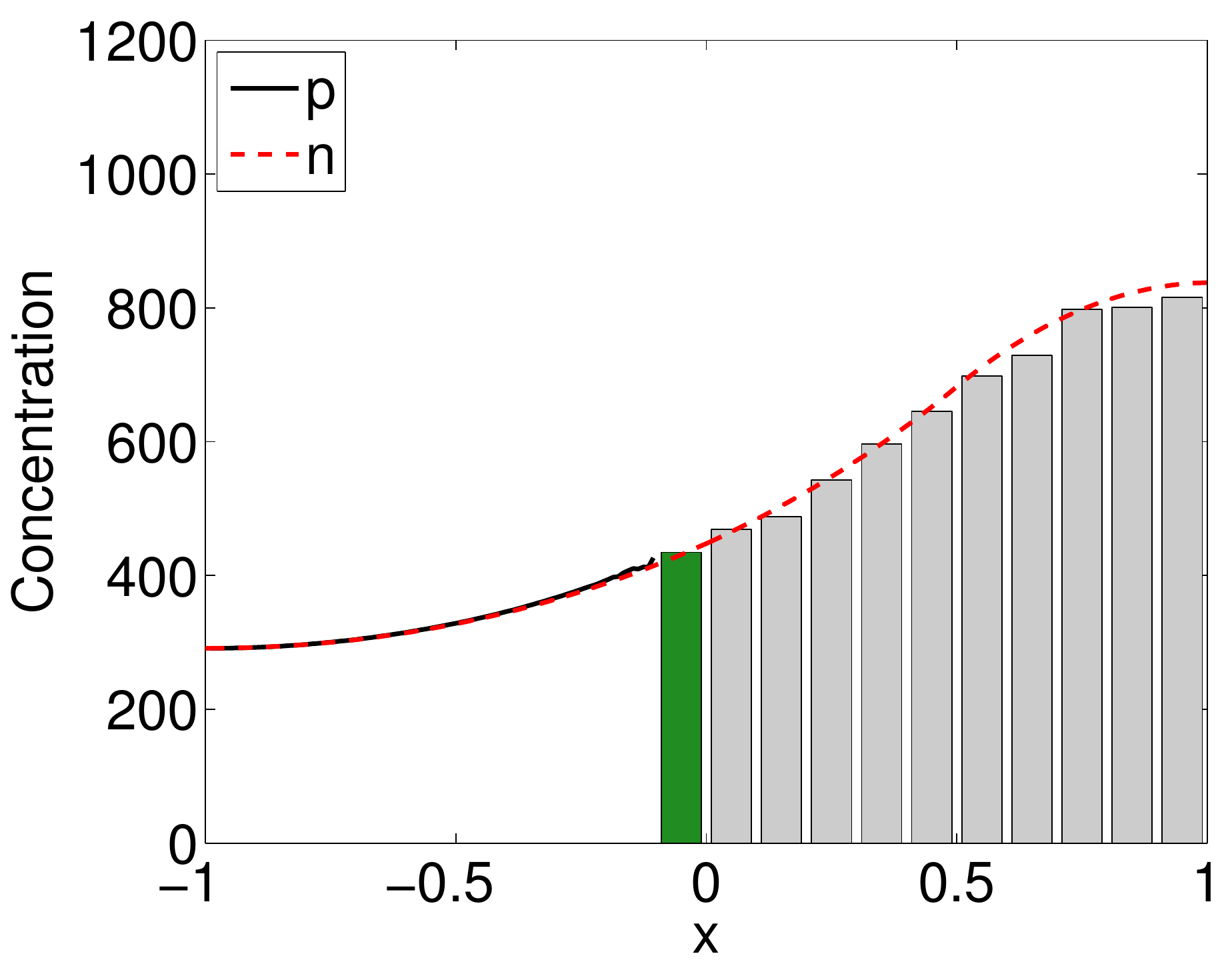, width=0.45\textwidth}
\label{fig:exp2:avg100}
}
\caption{Simulation results for Example 2. 
Dashed (red) line: exact solution given by $(\ref{eq:exp2:full})$;
solid line: $p(x, T_{\mathrm{final}})$ in $\Omega_P \setminus O$;
(gray) bars: spatial concentration of particles
at $t=T_{\mathrm{final}}$ in $\Omega_B \setminus  O$; 
(green) bar: $N_O(T_{\mathrm{final}})$ given by $(\ref{numberoverlap})$.
Parameters as described in the text.}
\label{fig:exp2}
\end{figure}        

\subsection{Example 3: chemisorption}
\label{secchemisorption}
Our last example is the polymer coating of a virus 
surface \cite{Fisher:2001:PCA,Erban:2007:DPI}.
We will describe it as irreversible adsorption (chemisorption)
of polymers to a two-dimensional surface as was introduced in
\cite{Erban:2007:TSR}. This example presents a typical application
area of PBD algorithms. A detailed model is used close to the
reactive boundary where positions of individual molecules influence 
the dynamics of diffusion-driven adsorption. On the other hand,
a less detailed model can be used far away from the adsorbing
surface. In the bulk the behaviour of reactive polymers
can be described by the macroscopic reaction-diffusion PDE 
(\ref{eq:reactdiff}) in the form
\begin{equation} \label{eq:RDPDE:degradation}
\frac{\partial p}{\partial t} 
= 
D \frac{\partial^2 p}{\partial z^2} - k_d \, p\,.
\end{equation}
in the semi-infinite domain $\Omega = (0, \infty)$ (here, $z$ is the 
distance from the reactive surface which is at $z=0$). This equation
takes into account two processes which mainly influence chemisorption
dynamics \cite{Erban:2007:TSR}: diffusion of polymer molecules and 
the hydrolysis of reactive groups in the solution. Both processes
can be implemented in the BD context as we saw in the previous
examples. However, this level of detail is only needed close to the
virus surface.

Whenever a polymer molecule interacts with the surface, it is either
reflected or (irreversibly) adsorbed. The chemisorption is modelled
by a random sequential adsorption (RSA) algorithm \cite{Erban:2007:CPS}: 
we check whether the corresponding binding site on the surface is free 
and then the reaction occurs with a certain probability. This probability 
is related to the reaction rate constant of the binding reaction as given
in \cite{Erban:2007:RBC}. The reader can find more details about
the model in \cite{Erban:2007:TSR}. In this paper, we show that
the PBD algorithms can be used to compute the results from 
\cite{Erban:2007:TSR}. We will use the same 
parameters, namely $D = 5\times 10^{-5}\, \mbox{mm}^2\, \mbox{s}^{-1}$,
$k_d  = 1.3\times 10^{-4}\, \mbox{s}^{-1}$ and $\Delta t = 0.01\, \mbox{s}$.
Then the mean displacement per time step according to
\eqref{eq:BD} is $\sqrt{2D \Delta t} = 10^{-3}\, \mbox{mm}$ and we
therefore choose the size of the overlap region of the PBD algorithm
(B1)--(B5) as $|O| = 10^{-2} \, \mbox{mm}$. 
From the results in \cite{Erban:2007:TSR},
we estimate that a maximum length of $L=2\, \mbox{mm}$ is enough 
to simulate the binding process and use the Dirichlet boundary condition
\begin{equation*}
n(L,t) = c_0 \exp(-k_d t)\,,
\end{equation*}
where $c_0 = 1.2\times 10^4\, \mbox{molecules}/\mbox{mm}$ is the initial
concentration of molecules (i.e. $n(z,0) \equiv c_0$ for $z \in \Omega$).
To apply the PBD algorithm (B1)--(B5), we choose
$\Omega_B = [0, 1.01)\, \mbox{mm}$, $O = (1, 1.01)\, \mbox{mm}$ and 
$\Omega_P = (1, 2) \, \mbox{mm}$. The RSA algorithm was performed using 
a nearest neighbour exclusion on a $100\times 100$ 
grid of receptor binding positions on the surface \cite{Erban:2007:CPS}.

In Figure~\ref{fig:exp4:c} we plot the concentration profile inside 
$\Omega_B$ of a single simulation at two different times (gray histograms). 
We compare the results of the PBD algorithm
(B1)--(B5) with the results of the RSA-PDE model presented in 
\cite{Erban:2007:TSR} (black lines). As shown in \cite{Erban:2007:TSR},
the RSA-PDE model also compares well with the full BD simulation. 
The number of molecules which are attached to the surface as a function
of time is plotted in Figure~\ref{fig:exp4:tbind} (six realisations
of the PBD algorithm (B1)--(B5) are plotted as green solid lines). Again, 
we see an excellent agreement with the result from \cite{Erban:2007:TSR}
which is plotted as the black dashed line.
\begin{figure}
        \center
        \subfigure[$T_{\mathrm{final}} = 20 min$]{
            \epsfig{file=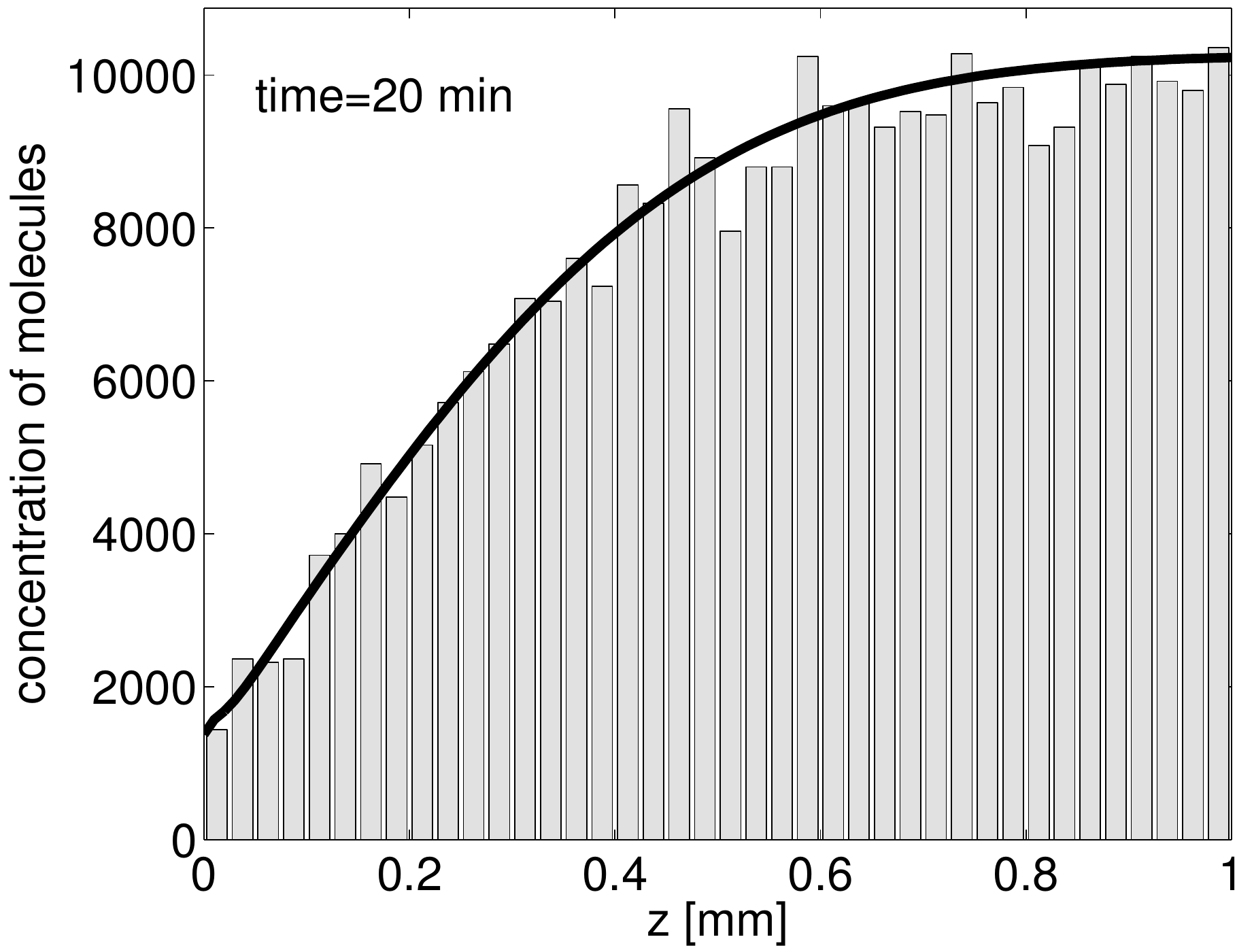, width=0.45\textwidth}
            \label{fig:exp4:20min}
        }
        \subfigure[$T_{\mathrm{final}} = 80 min$]{
            \epsfig{file=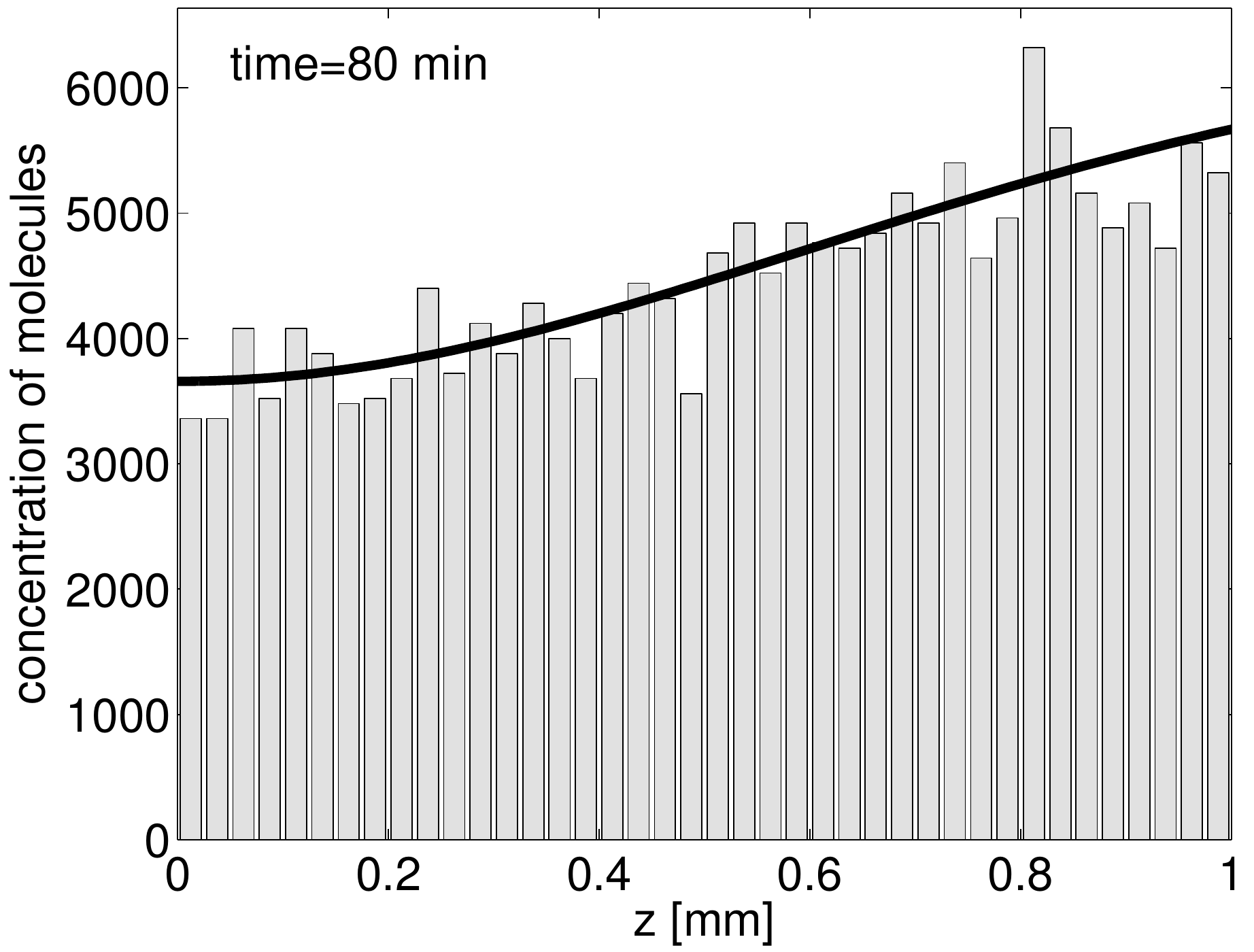, width=0.45\textwidth}
            \label{fig:exp4:80min}
        }
        \caption{Example 3 (chemisorption to virus surface). 
	(Gray) histograms:
	concentration profile in molecules/mm at two given times
	computed by the PBD algorithm (B1)--(B5);
	solid line: results of the
	RSA-PDE model presented in \cite{Erban:2007:TSR}.
	Parameters as shown in the text. }
    \label{fig:exp4:c}
\end{figure}

\begin{figure}
        \center
        \epsfig{file=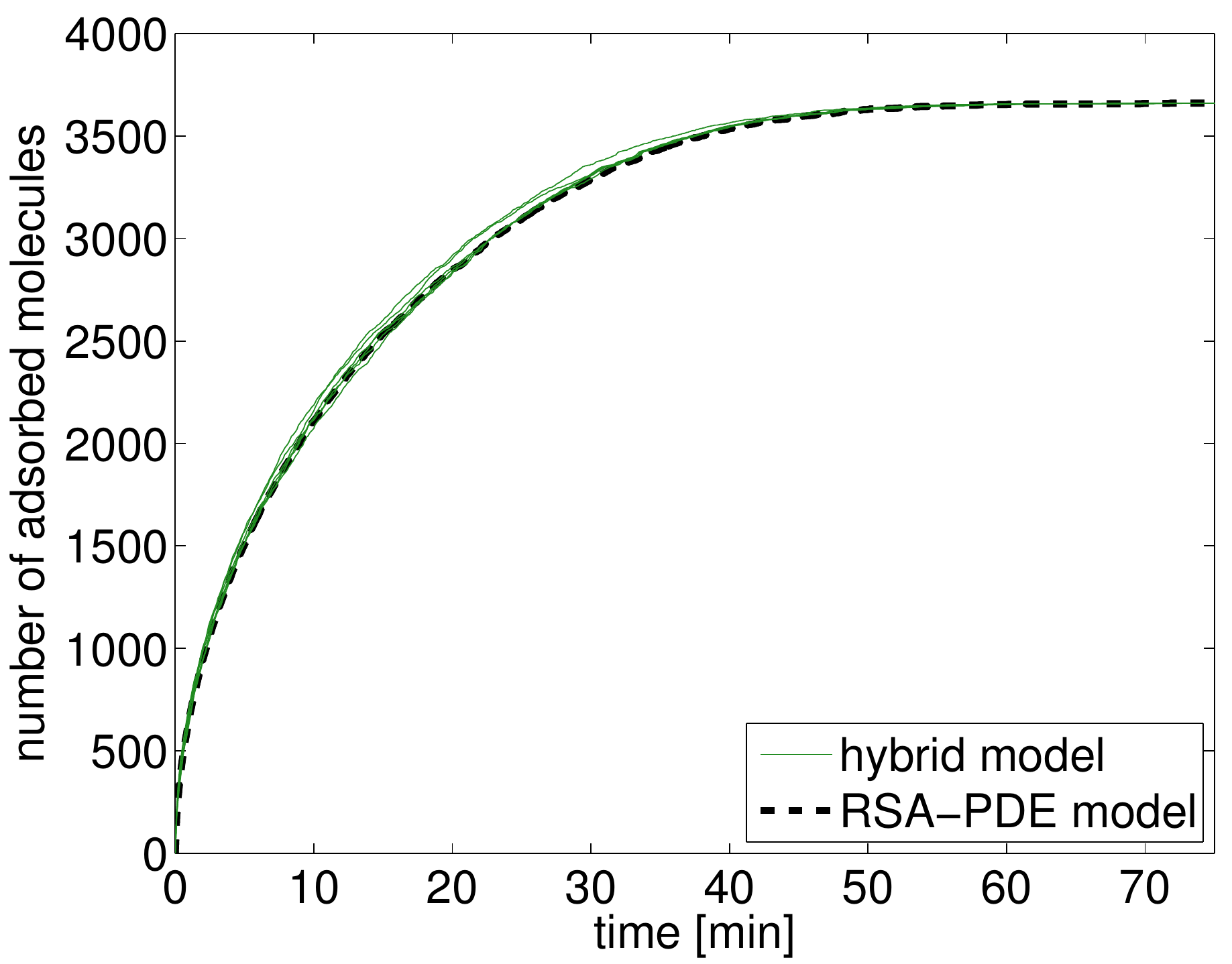, width=0.45\textwidth}
        \caption{Example 3 (chemisorption to virus surface).
	Number of polymer molecules which are bound to the virus 
	surface as a function of time.  
	(Green) solid lines: six realisations computed by the 
	PBD algorithm (B1)--(B5); (black) dashed line: results of the
	RSA-PDE model presented in \cite{Erban:2007:TSR}.
	Parameters as shown in the text.}
    \label{fig:exp4:tbind}
\end{figure}

\section{Discussion} \label{sec:discussion}
In this paper we have presented two PBD algorithms that combine 
Brownian dynamics with mean-field reaction-diffusion PDEs. 
This method produces exact Brownian dynamics simulations 
in one part of the domain and couples them with mean field 
approximations in another part of the domain. An algorithm 
of this type is useful for various application areas in
computational biology and beyond, for example, when a detailed
description of individual molecules is required near a receptor or
ion channel, but becomes impractical in the bulk of a cell
\cite{Corry:2000:TCT,Flegg:2012:DSN}; or when a detailed stochastic 
simulation of actin dynamics is required inside filopodia, but
becomes impractical in the bulk of a cell \cite{Zhuravlev:2009:MNC}.
Another application area, chemisorption, was discussed in Section
\ref{secchemisorption}. By using our approach it would also be 
possible to use finite-sized particles in the BD simulation and
couple these with the corresponding mean field results presented
in \cite{Bruna:2012:EVE}. 

In the literature several hybrid models have been developed in the
context of fluid dynamics, but they do not discuss issues that arise 
from the incorporation of chemical reactions. Alexander et 
al \cite{Alexander:2002:ARS} presents a hybrid model that uses virtual 
particles at the grid point nearest to the interface to calculate fluxes 
across the boundary and to generate accurate density fluctuations inside 
the particle region. Reference \cite{Wagner:2004:HCF} extends this 
approach by the introduction of an overlap region similar to $O$ 
introduced in Section~\ref{sec:probation}. An identical flux 
exchange with particles confined to a grid is presented in 
\cite{Flekkoy:2001:CPF}. Chemical reactions in the solution
were considered in \cite{Erban:2007:TSR}. This model was discussed
in Section \ref{secchemisorption}. It couples Brownian dynamics
of molecules in the solution with a more detailed description of the
adsorbing boundary. In \cite{Erban:2007:TSR} a hybrid (RSA-PDE) model 
has been developed which replaces BD in the solution by solving the PDE
(\ref{eq:RDPDE:degradation}) with a suitable stochastic boundary
condition. The PBD algorithms are able to replace the stochastic
boundary condition by a (small) BD region close to the surface. Although
the hybrid RSA-PDE model introduced in \cite{Erban:2007:TSR} was
sufficient in the case of (irreversible) adsorption, the situation
is becoming more challenging whenever the binding reaction is
reversible \cite{Lipkova:2011:ABD}. In this case, a molecule
which is released by the surface will initially stay close to the 
surface and can rebind to the same receptor (binding site). This 
geminate recombination can be captured by the PBD approach. Reversible 
reactions are common in biological applications 
\cite{Lipkova:2011:ABD,Flegg:2012:DSN}.

Hybrid approaches for reaction-diffusion processes which couple 
different modelling approaches have also been introduced in the
literature \cite{Moro:2004:HMS,Flegg:2012:TRM,Engblom:2009:SSR,Ferm:2010:AAS}. 
A mesoscopic lattice-based description coupled with macroscopic 
Fisher-Kolmogorov-Petrovsky-Piscounov PDE was
used in \cite{Moro:2004:HMS} to study front propagation
in a lattice-based reaction-diffusion model. A hybrid model for 
reaction-diffusion systems in porous media that combines 
pore-scale models with Darcy-scale models is presented 
in \cite{Tartakovsky:2008:HSR}. Flegg et al \cite{Flegg:2012:TRM}
introduced the so-called Two Regime Method which couples
a lattice-based (compartment-based) reaction-diffusion model 
with BD simulations. One advantage of the PBD algorithms over
the Two Regime Method is that 
there are more efficient tools for PDE simulations than for compartment-based
reaction-diffusion models. On the other hand, compartment-based
models provide more details (including fluctuations) and
hybrid models which couple BD simulations with compartment-based
models do not require the overlap region \cite{Flegg:2012:TRM,Ferm:2010:AAS}.
Since it is possible to couple the macroscopic PDE description
with (mesoscopic) compartment-based models \cite{Engblom:2009:SSR} 
and compartment-based models with (microscopic) BD simulations
\cite{Flegg:2012:TRM}, then an alternative approach to PBD algorithms
would be to use compartment-based models in the overlap region. That is,
the computational domain would be divided into three regions where
the PDE, compartment-based and BD descriptions would be used. These
three regimes would be coupled using the results from the literature
\cite{Engblom:2009:SSR,Flegg:2012:TRM}. Compartment-based models and 
macroscopic PDEs can also be coupled through another intermediate 
regime using a tau-leap method \cite{Ferm:2010:AAS}. In this paper,
we showed that PDE models and BD simulations can be coupled without
using intermediate compartment-based models.

The PBD algorithms presented in this paper should be seen as a first 
step towards a more general setting. Some parts of the algorithm 
(B1)--(B5) extend easily (at least theoretically) into higher 
dimensions, but in practice additionally difficulties are posed. 
One example is the necessity to sample from a multidimensional 
probability distribution to find the position of newly created molecules.
Additionally, in higher dimensions one can also expect to deal 
with higher order reactions, including bimolecular reactions. 
For a discussion of how to implement bimolecular reactions for 
BD simulations, we refer to \cite{Erban:2009:SMR}, but the real 
problem occurs inside the overlap region $O = \Omega_B \cap \Omega_P$, 
where a molecule could react with another molecule, or with the 
continuum. Since the reaction-diffusion PDEs are solved numerically
using a suitable mesh, it is important to study methods for coupling 
individual molecules with the numerical discretization of macroscopic 
PDEs \cite{Franz:2012:HMI}. For the ion channel application mentioned 
before, one also needs to think about how to incorporate 
electrical charges and resulting forces into the system. One can 
imagine that these forces act as a boundary condition on the 
continuum model and as an effective force on the particles.

\section*{Acknowledgments}
The research leading to these results has received funding from the
European Research Council under the \emph{European Community's} Seventh
Framework Programme \emph{(FP7/2007-2013)} / ERC \emph{grant agreement} No. 239870.
This publication was based on work supported in part by Award
No KUK-C1-013-04, made by King Abdullah University of Science and
Technology (KAUST). Radek Erban would also like to thank the 
Royal Society for a University
Research Fellowship; Brasenose College, University of Oxford, for a
Nicholas Kurti Junior Fellowship and the Leverhulme Trust for a Philip
Leverhulme Prize.

\bibliographystyle{siam}

\end{document}